\documentclass[aps,pre,twocolumn,showpacs,superscriptaddress]{revtex4-1}

\usepackage{graphicx}
\usepackage{amsmath,amsfonts,amsthm,amssymb}
\usepackage{bm}

\usepackage{color}
\usepackage{epstopdf}

\usepackage{hyperref}
\hypersetup{
    colorlinks,
    citecolor=blue,
    filecolor=blue,
    linkcolor=blue,
    urlcolor=blue
}

\usepackage[section] {placeins}

\definecolor{commentaqua}{rgb}{0.2, 0.6, 0.9}
\definecolor{commentred}{rgb}{0.8, 0.2, 0.2}

\newtheorem{definition}{Definition}
\newtheorem{theorem}{Theorem}
\newtheorem{proposition}{Proposition}
\newtheorem{corollary}{Corollary}
\newtheorem{lemma}{Lemma}

\newcommand{\uvec}{\mathbf{u}}
\newcommand{\uvecperp}{\mathbf{u}^{\perp}}
\newcommand{\rvec}{\mathbf{r}}
\newcommand{\rdot}{\dot{\mathbf{r}}}

\newcommand{\ghat}{\hat{\mathbf{g}}}
\newcommand{\nhat}{\hat{\mathbf{n}}}
\newcommand{\drdlambda}{d\mathbf{r} / d\lambda}

\DeclareMathOperator{\sgn}{sgn}

\begin{document}

\title{Frozen reaction fronts in steady flows: a burning-invariant-manifold perspective}

\author{John~R.~Mahoney}
\email{jmahoney3@ucmerced.edu}
\affiliation{University of California, Merced, California 95344, USA}
\author{John~Li}
\affiliation{University of California, Merced, California 95344, USA}
\affiliation{University of Southern California, Los Angeles, CA 90089, USA}
\author{Carleen~Boyer}
\affiliation{Bucknell University, Lewisburg, Pennsylvania 17837, USA}
\author{Tom~Solomon}
\affiliation{Bucknell University, Lewisburg, Pennsylvania 17837, USA}
\author{Kevin~A.~Mitchell}
\email{kmitchell@ucmerced.edu}
\affiliation{University of California, Merced, California 95344, USA}

\date{\today}

\begin{abstract}
The dynamics of fronts, such as chemical reaction fronts, propagating in
two-dimensional fluid flows can be remarkably rich and varied.  For
time-invariant flows, the front dynamics may simplify, settling in to a
steady state in which the reacted domain is static, and the front
appears ``frozen''.  Our central result is that these frozen fronts in the
two-dimensional fluid are composed of segments of \emph{burning invariant
manifolds}---invariant manifolds of front-element 
dynamics in $xy\theta$-space, where $\theta$ is the front orientation.
Burning invariant manifolds (BIMs) have been identified
previously as important local barriers to front propagation in fluid
flows.   The relevance of BIMs for frozen fronts rests in their ability,
under appropriate conditions, to form global barriers, separating
reacted domains from nonreacted domains for all time.  The second main
result of this paper is an understanding of bifurcations that lead from
a nonfrozen state to a frozen state, as well as bifurcations that change
the topological structure of the frozen front.  Though the primary
results of this study apply to general fluid flows, our analysis focuses
on a chain of vortices in a channel flow with an imposed wind.  For this
system, we present both experimental and numerical studies that support
the theoretical analysis developed here.
\end{abstract}

 % Chemically reactive flows, Pattern formation in reactions with diffusion, flow and heat transfer, Dynamical systems methods

\maketitle
%\tableofcontents

%%%%%%%%%%%%%%%%%%%%%%%%%%%%%%%%%%%%%%%%%%%%
\section{Introduction: Reacting flows and frozen fronts}
\label{sec:intro}

The evolution of an autocatalytic reaction $A+B \to 2A$ in a
spatially-extended system is characterized by the propagation of
reaction fronts that separate the species $A$ and $B$. The motion of
these fronts is well-understood for reaction-diffusion systems in
the absence of any substrate flow. The effects of fluid motion on
fronts in the more general {\em advection-reaction-diffusion} system
have only recently received significant attention.  This is somewhat
surprising, given the applicability of advection-reaction-diffusion to
a wide range of systems, including microfluidic chemical reactors
\cite{Mezic07}, plasmas \cite{beule98}, the dynamics of ecosystems in
the oceans (e.g., plankton blooms) \cite{scotti07}, cellular- and
embryonic-scale biological processes \cite{prigogine84, babloyantz86},
and the propagation of diseases in society \cite{russell04}. It has
been recently proposed that the motion of reaction fronts in fluid
flows may be dominated by the presence of {\em burning invariant
  manifolds} (BIMs), which act as one-way barriers to advancing fronts
\cite{Mahoney12, Mitchell12b}. The existence of BIMs and their
function as one-way barriers has been verified experimentally in
time-independent and time-periodic vortex chain flows, as well as 2D
disordered vortex flows \cite{Bargteil12}.

Experiments have shown that reaction fronts tend to pin to vortex
structures in the presence of an imposed wind \cite{Schwartz08}.
These fronts neither propagate forward against the wind nor are blown
backwards, but remain ``frozen''. This behavior is surprisingly
robust, occurring over more than an order of magnitude of wind speeds
and a variety of underlying flows ranging from confined vortex chains
to extended, spatially-random flows. Figure
\ref{fig:converge_to_pinned} shows a sequence of images from
experiments showing the evolution of a triggered, autocatalytic
reaction front in a vortex chain with wind. The front eventually
stabilizes and remains fixed for the duration of the experiment.

\begin{figure}[bt]
\includegraphics[width=\linewidth]{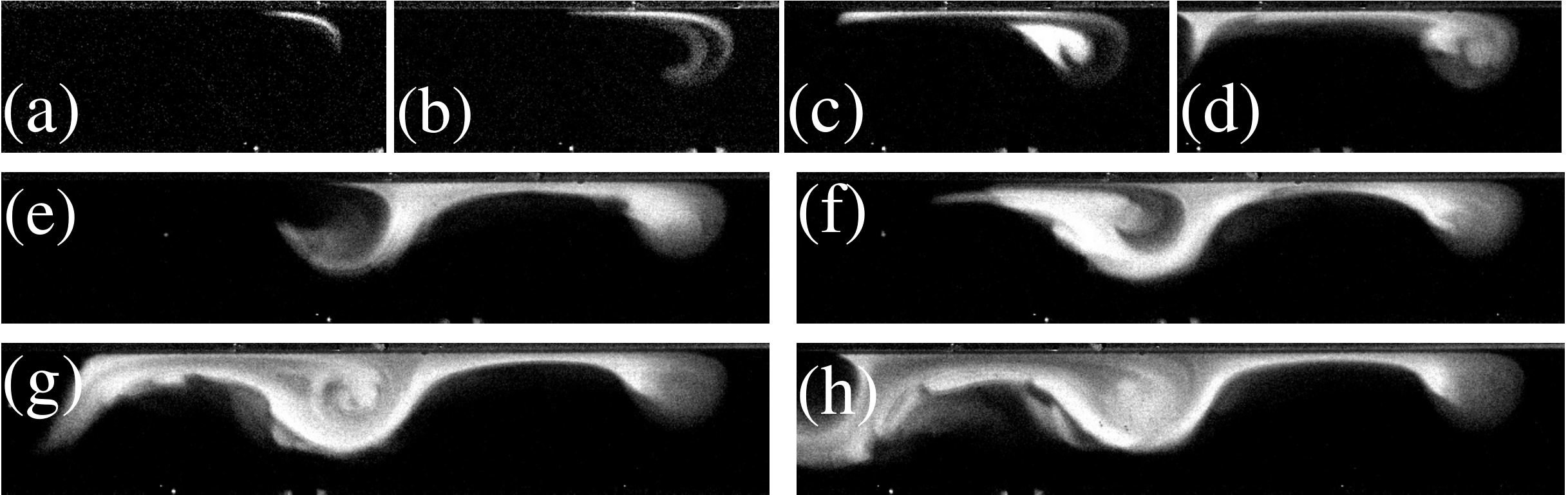}
\caption{Sequence showing the evolution of a triggered reaction front.
  The maximum fluid vortex speed (in the absence of wind) is $U =
  1.4$~mm/s, and the wind speed is $V_w = 0.90$~mm/s.  The images in
  the sequence are separated by $10$~s.}
\label{fig:converge_to_pinned}
\end{figure}

In this paper, we use the theory of BIMs to characterize these
\emph{frozen fronts} (FFs). FFs occur when a BIM spans the entire
width of the system with no changes in blocking direction, or when
there is a combination of overlapping BIMs, with the same blocking
directions, that together span the system. In either of these
situations, the shape of the FF is determined by the shape of the BIMs
responsible.  We illustrate the creation of FFs and changes in their
structure by increasing the wind applied to a canonical base flow (the
alternating vortex chain) with a propagating chemical reaction.  We
present both experimental and numerical studies of this system.

This paper is organized as follows.  We begin in
Sec.~\ref{sec:experiments} by presenting experiments involving
reaction fronts in a particular quasi-two-dimensional fluid flow---the
``windy alternating vortex chain''.  The images in this section
illustrate the behavior of FFs under an imposed wind of various
strengths.  Section~\ref{sec:BIM} recalls some basic aspects of
\emph{burning invariant manifolds} (BIMs)---geometric structures that
govern the progress of fronts in fluid flows---including the
three-dimensional dynamics of front elements and fixed points of this
system.  Next, Sec.~\ref{sec:FF} connects the previous two sections by
showing that FFs are composed of BIM segments.
Section~\ref{sec:theory} considers FFs in a numerical model of the
experimental flow.  It parallels Sec.~\ref{sec:experiments} by
increasing the applied wind and observing the resulting changes in the
FFs.  Here we discuss in detail the various FF topologies and the
dynamical systems mechanism underlying the transitions which connect
them.  There are four appendices.  Appendix~\ref{app:sliding_fronts}
introduces a two-dimensional invariant surface of ``sliding fronts'',
which is used to prove several key results in the paper.
Appendix~\ref{app:stability} establishes the stability condition that
frozen fronts must satisfy.  Appendix~\ref{app:far_field} provides
some technical analysis concerning the structure of frozen fronts at
infinity.  Finally, Appendix~\ref{app:SSS} examines attracting fixed
points for this system.

%%%%%%%%%%%%%%%%%%%%%%%%%%%%%%%%%%%%%%%%%%%%
\section{Experiments: windy alternating vortex chain flow}
\label{sec:experiments}

%Draft of text for first paragraph of Section II
%-----------------------------------------------
The alternating vortex chain fluid flow has been the subject of much
study, both theoretical and experimental. It has been used as a model
of a two-dimensional cross-section of Rayleigh-B\'enard (thermal)
convection \cite{chandrasekhar,busse1974,busse1986,Solomon88}
and Taylor-Couette vortices \cite{hohenberg93}, and can be used to
model vortex chains and streets in oceanic and atmospheric flows
\cite{atmos_convection1993,atmos_convection1996}.  The
alternating vortex chain has been used to study enhancement
of long-range, fluid transport in cellular
flows \cite{shraiman87,solomon88b,Camassa91,Solomon96}.
More recently, it has been used repeatedly in studies of
chemical front propagation in advection-reaction-diffusion systems
\cite{Abel01,Abel02,Cencini03,Paoletti05,Paoletti05b,Mahoney12}.  Here we
modify this flow by adding a uniform ``wind'', creating the {\em windy
alternating vortex chain} \cite{Schwartz08}.

%%%%%%%%%%%%%%%
\subsection{Experimental setup}

\begin{figure}[bt]
\includegraphics[width=\linewidth]{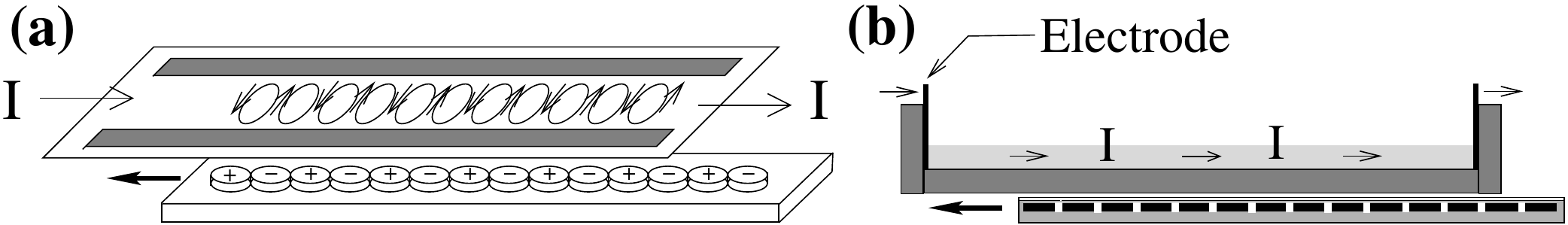}
\caption{Experimental apparatus. (a) Exploded view of alternating
  fluid vortices above array of magnets. Current through fluid induces
  Lorenz force.  (b) Side view of the apparatus.  A 2 mm thick layer
  of an electrolytic fluid is contained in an acrylic box.  The chain
  of Nd-Fe-Bo magnets moves on a translation stage below the box. }
\label{fig:experimental_apparatus}
\end{figure}

The experiments are conducted in a quasi-2D flow composed of a chain 
of vortices in a thin (2 mm) fluid layer.
The flow is produced using a magnetohydrodynamic forcing technique, as
shown in Fig.~\ref{fig:experimental_apparatus}.
A chain of permanent 1.9cm-diameter Nd-Fe-Bo magnets sits below the
fluid layer, thereby imposing a spatially-varying magnetic field.  An
electric current is passed though this electrolytic fluid, generating
Lorentz forces on the fluid.
%An electrical current passing through
%an electrolytic solution interacts with a spatially-varying
%magnetic field produced by a chain of permanent Nd-Fe-Bo magnets
%below the fluid layer. 
%
In conjunction with rigid, plastic side-walls that bound the region of
interest, the result is an alternating chain of well-controlled
vortices. The magnets are mounted on a translation stage; motion of
the translation stage results in motion of the magnets and,
consequently, the fluid vortices. In these experiments, we move the
magnets (and the vortices) with a constant speed $V_w$.  In the
reference frame moving with the magnets, the flow is a stationary
chain of vortices with an imposed, uniform wind of speed $V_w$.

The fronts are produced in the experiments with the excitable,
ferroin-catalyzed Belousov-Zhabotinsky chemical reaction
\cite{Boehmer08, scott94}.  At the beginning of an experimental run,
the ferroin indicator in the solution is in its reduced (orange)
state.  A reaction is then triggered by briefly dipping a silver wire
into the fluid.  The silver oxidizes the ferroin in its vicinity,
changing the local indicator to a blue-green color. The oxidized
indicator in turn oxidizes the ferroin of {\em its} neighbors,
resulting in a blue-green reaction front that steadily propagates
outward from the trigger point with a roughly constant propagation
speed $V_0$.  For all experiments presented in this article, the
propagation speed is $V_0 = 0.07$~mm/s.  The front is a pulse-like
front---behind the leading edge of the front, the reaction relaxes
back to its reduced (orange) state and can be re-triggered. Previous
studies \cite{Paoletti05, Paoletti05b, Cencini03} have indicated that
the behavior of the leading edge of these pulse-like fronts in a fluid
flow is identical to the behavior of the leading edge of a burn-type
reaction.  (Burn-type reactions do not relax back, rather $A + B \to
2A$ and stays that way.)

%%%%%%%%%%%%%%%

%%%%%%%%%%%%%%%
\subsection{Experimental results}

We focus on the behavior of the leading edge of the reaction front
that propagates \emph{against} the imposed wind.  (In the lab frame,
these fronts propagate in the direction of the imposed motion of the
vortex cores.)  An example of a typical experiment is shown in
Fig.~\ref{fig:both_ref_frames}.  As viewed in the laboratory reference
frame (Fig.~\ref{fig:both_ref_frames}a), the front continually
propagates in both directions; in the reference frame moving with the
vortices (Fig.~\ref{fig:both_ref_frames}b) the right-most edge of the
reaction front converges to a steady-state stationary shape that
remains fixed for the duration of the run. From here on, we use the
expression ``wind'' $V_w$ to refer to either the translational speed
of the vortices in the lab frame or the speed of the uniform wind in
the vortex reference frame.

%$V_w$ \TS{CHECK NOTATION -- we typically scale all velocities by $V_0$.  Is
%that the convention being used through the rest of the paper, or are velocities
%scaled by the maximum vortex velocity?} = 0.30 mm/s, $v_w = 4.3.$ 

\begin{figure}[bt]
\includegraphics[width=\linewidth]{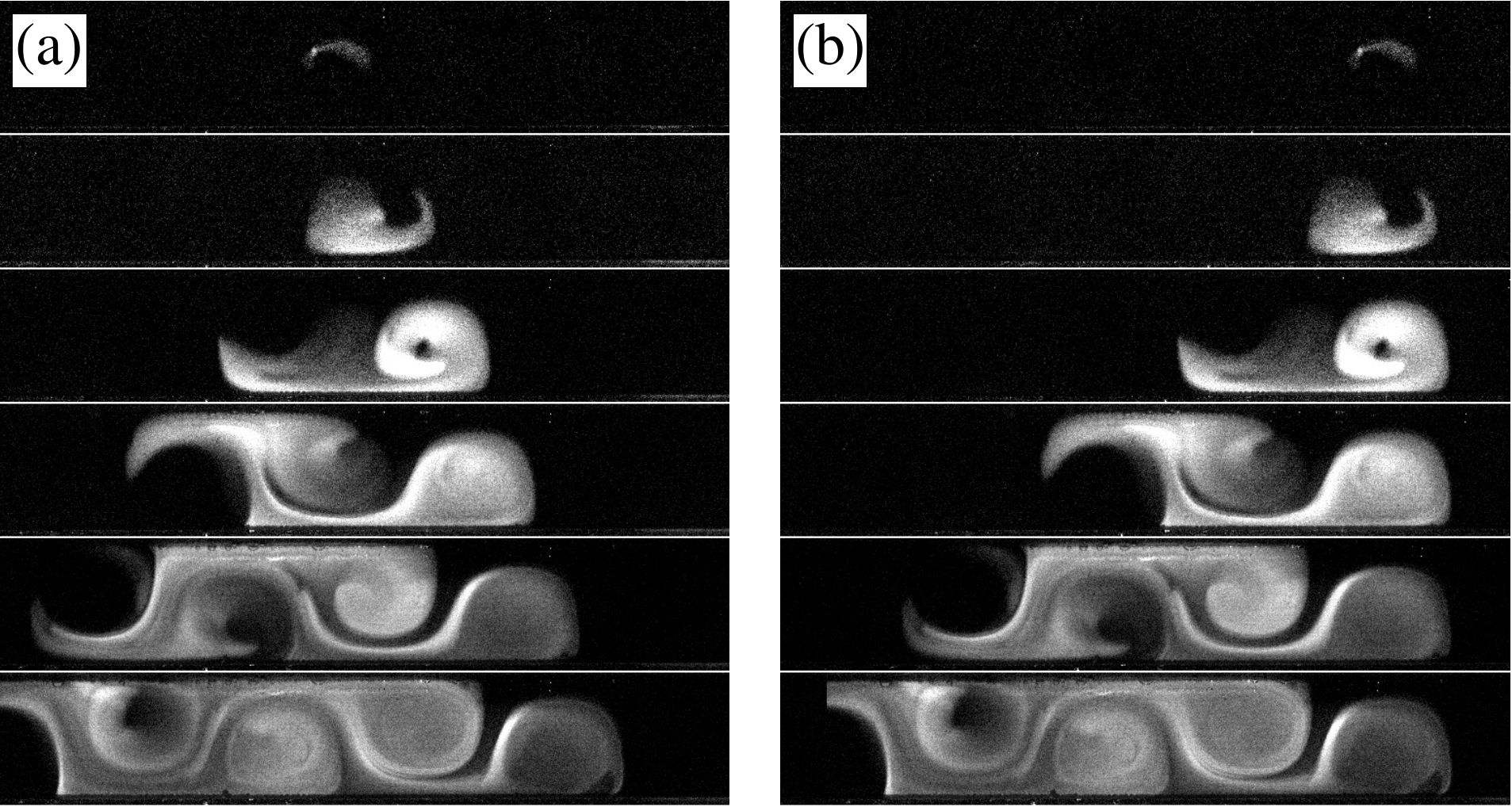}
\caption{Sequences showing the evolution of a reaction front in a
  vortex chain. (a) Lab frame, with the vortices moving to the right.
  (b) Reference frame moving with the vortices.  In this frame, the
  vortices are stationary and a wind blows across the vortices toward
  the left.  $U = 1.4$~mm/s, $V_w=0.30$~mm/s.  Images in the sequences
  are separated by $20$~s.}
\label{fig:both_ref_frames}
\end{figure}

The propagation of a reaction front in the alternating vortex flow in
the absence of an imposed wind has been discussed in detail in
previous papers \cite{Paoletti05, Paoletti05b, Pocheau06, Pocheau08}.
The reaction front is carried around each vortex with the flow and
``burns'' across the separatrix from one vortex to the next, resulting
in long-range propagation that is significantly faster than the
reaction-diffusion speed $V_0$ in a static fluid.  The long-term
average front speed is independent of the initial stimulation.

If a uniform wind $V_w < V_0$ is applied (i.e., the wind speed is
smaller than the reaction-diffusion speed), the reaction front still
propagates to the right against the wind, although the long-range
propagation speed is reduced.  At $V_w = V_0$, there is a transition
where the front neither advances against the wind nor is blown
backwards \cite{Schwartz08}.  Figure \ref{fig:onset_of_pinning_exp}
shows a sequence for a reaction front triggered in a flow with wind
$V_w$ just below $V_0$
(Fig.~\ref{fig:onset_of_pinning_exp}a-\ref{fig:onset_of_pinning_exp}c)
and $V_w$ just above $V_0$
(Fig.~\ref{fig:onset_of_pinning_exp}d-\ref{fig:onset_of_pinning_exp}f).
%
% Figure that shows front evolution just before and just after critical wind.
% Optionally, could use a variation of small_wind.eps instead.
\begin{figure}[bt]
\includegraphics[width=\linewidth]{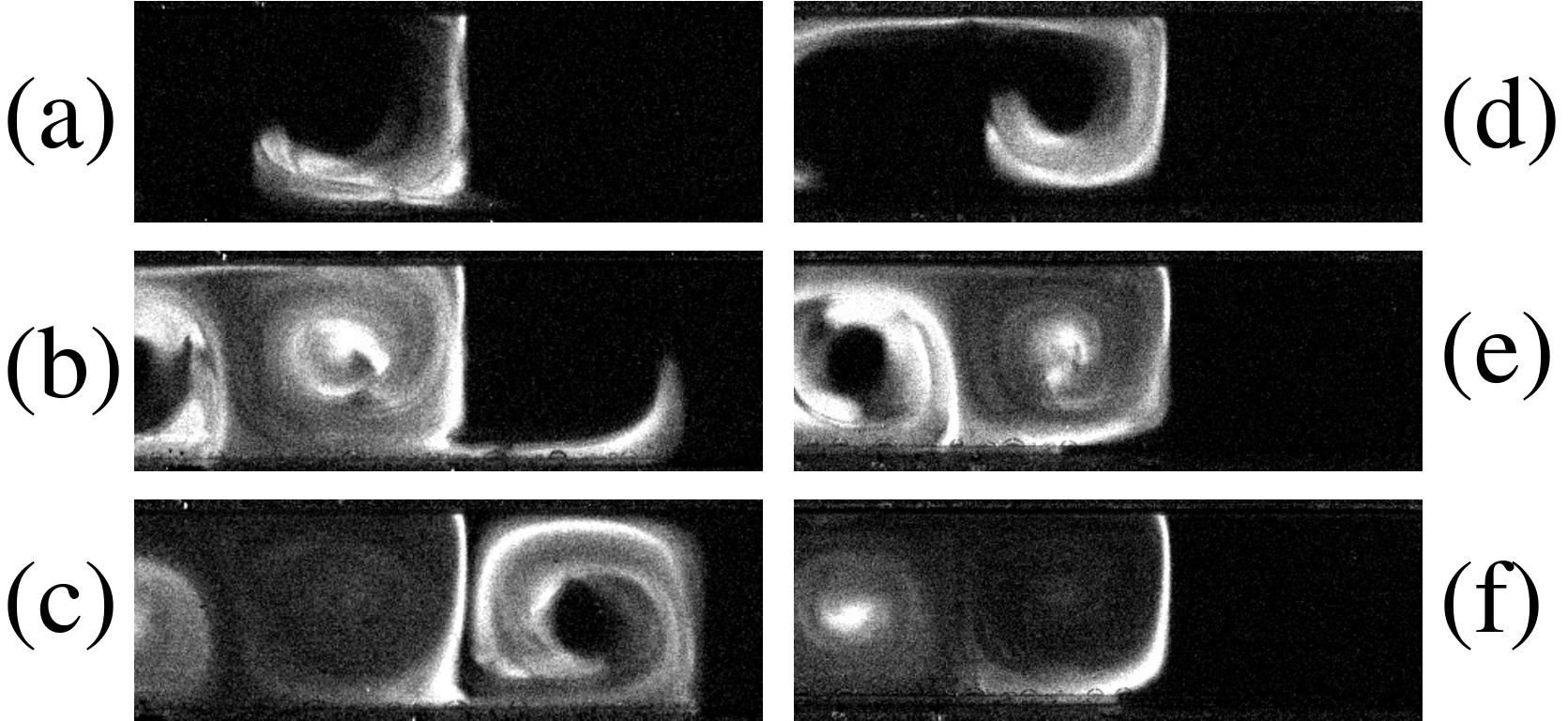}
\caption{Two sequences demonstrate front evolution near critical wind
  speed.  The maximum fluid vortex speed (in the absence of wind) is
  $U = 0.7$~mm/s. For wind value $V_w = 0.085$ mm/s, the front is (a)
  nearly vertical, (b) finds a small passage into the next right
  vortex, and (c) fills in the right vortex continuing down the
  channel. For wind value $V_w = 0.090$ mm/s, the front is (d) nearly
  vertical, (e) does not find passage to the right, and (f) remains
  unchanged from (e)---a frozen front. The time between images is
  $40$~s in both sequences.  Note that there is a small amount of
  experimental noise that increases the transition slightly above
  $V_w = V_0 = 0.07$~mm/s.}
\label{fig:onset_of_pinning_exp}
\end{figure}
The shape of the FF is not arbitrary; rather a wide range of initial
stimulations will result in fronts that converge onto the same
structure.  For $V_w = V_0$, the shape of the FF corresponds well with
the advective separatrix having $V_w = 0$.

As the strength of the imposed wind is increased, the shape of the FF
evolves considerably.  Figure \ref{fig:pinfronts3} shows time-averaged
images of the steady-state reaction fronts for several different wind
speeds.  With increasing wind speed, the contact point of the FF with
the upper boundary does not move much.  There is also a shift-flip
symmetry apparent in Fig.~\ref{fig:pinfronts3}b-\ref{fig:pinfronts3}g;
for every FF originating from a contact point there is a flipped
version of the same structure originating from a contact point one
vortex width leftward. Consequently, for any wind speed, the leading
edge of the front could be pinned to any one of these contact points;
i.e., any FF could be replaced by the same shape, shifted by one
vortex and flipped vertically.

\begin{figure}[bt]
\includegraphics[width=\linewidth]{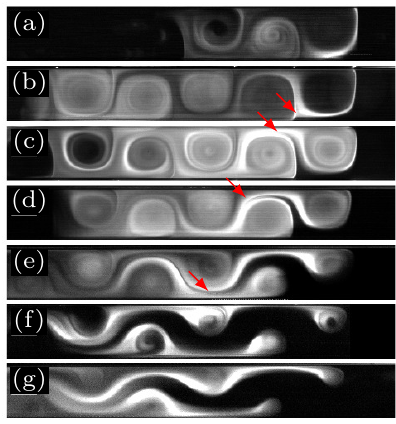}
\caption{(Color online.)  Time-averaged images of steady-state
  reactions for several wind speeds.  $U = 1.4$~mm/s for all.  $V_w =$
  (a) 0.15~mm/s, (b) 0.16~mm/s, (c) 0.20~mm/s, (d) 0.30~mm/s, (e)
  0.60~mm/s, (f) 0.90~mm/s, and (g) 1.2~mm/s.  Arrows indicate apparent
  discontinuities in the FF tangent direction. }
\label{fig:pinfronts3}
\end{figure}

The shift-flip symmetry is also relevant to a change in the structure
of the FF as the wind speed is increased.  The front develops a point,
or corner, with an apparently discontinuous derivative
(Fig.~\ref{fig:pinfronts3}b). This point moves leftward for larger and
larger wind speeds (Fig.~\ref{fig:pinfronts3}c-\ref{fig:pinfronts3}g).
This concave corner first appears near the downwind contact point (one
vortex width downwind in Fig.~\ref{fig:pinfronts3}b) and moves away
from the channel wall.  In this situation, the FF is composed of a
combination of smooth curves that originate at different contact
points.

Above a minimum wind speed, the shape of the FF is no longer uniquely
determined (modulo the flip-shift symmetry); rather, more than one
front shape is possible, depending on the manner in which the front is
triggered (Fig.~\ref{fig:diff_triggers}).  It is possible to trigger a
reaction front that pins only to the structure emanating from a single
contact point, as in Fig.~\ref{fig:diff_triggers}a.  But the same flow
allows for other FFs, such as in Fig.~\ref{fig:diff_triggers}b.  The
number of different possible FF shapes increases with the wind. As can
be seen in both Fig.~\ref{fig:pinfronts3} and
Fig.~\ref{fig:diff_triggers}, the front shapes are stretched out
significantly with increasing wind speed, spanning more and more
vortex cells.  For all except the smallest wind speeds, a FF can be
composed of structures pinned onto adjacent vortex contact points, as
in Figs.~\ref{fig:pinfronts3}b-\ref{fig:pinfronts3}g and
Figs.~\ref{fig:diff_triggers}b and \ref{fig:diff_triggers}c.  For
larger wind speeds, additional FF shapes are possible. As an example,
Fig.~\ref{fig:diff_triggers}d shows a FF composed of two structures
originating from contact points separated by 5 vortex widths.

Experimentally, the more complex steady-state front shapes are often
found by simultaneously triggering the reaction in multiple
locations. However, these complex shapes appear to be sometimes
accessible with even a single, well-placed trigger. A more detailed
theoretical treatment of these ``basins of attraction'' is in
preparation.

\begin{figure}[bt]
\includegraphics[width=\linewidth]{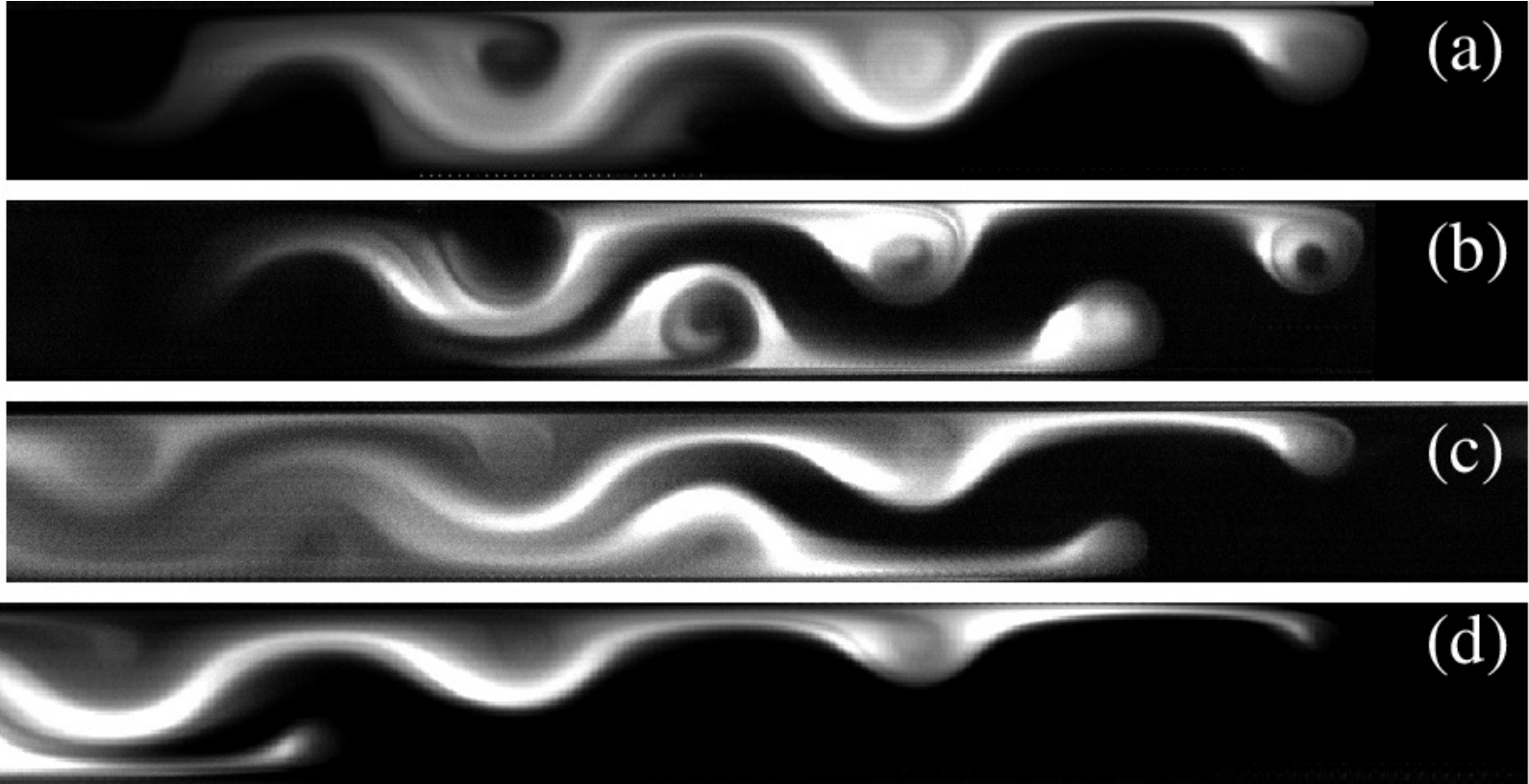}
\caption{Multiple FFs are realized with the same flow, depending on
  how the reaction is triggered.  $U = 1.4$~mm/s for all. $V_w$ is
  0.90~mm/s for (a,b) and 1.2~mm/s for (c,d).}
\label{fig:diff_triggers}
\end{figure}

For large enough wind, the stable state is lost completely, with the front
being ``blown backwards'' downwind.  A complete parameter space showing
the range of wind speeds for FFs can be found in Ref.~\cite{Schwartz08}.

%%%%%%%%%%%%%%%%%%%%%%%%%%%%%%%%%%%%%%%%%%%%
\section{BIM review}
\label{sec:BIM}

% This section reviews basic features of burning invariant manifolds drawing from \cite{Mahoney12, \cite{Mitchell12b}, Bargteil12}.

We model advection-reaction-diffusion systems, such as the above
experiments, by considering only the front.  This amounts to taking
the so-called ``sharp-front'', or geometric-optics limit.  While some
other studies have made use of a grid-based computational scheme
\cite{Abel02, Cencini03}, focusing on the front is numerically
economical and theoretically insightful.  By assuming that the front
progresses in a curvature-independent way~\cite{Note1}, 
%
%\footnote{There are small
%  variations in the propagation speed due to curvature of the reaction
%  front, but this effect is minimal in these experiments.} 
%
the front may be regarded as the collection of independent front
elements that comprise it.  Although not crucial to the basic ideas
here, we also assume that the ``burning speed''~\cite{Note2}
%
%\footnote{The experiment presented here depends on a redox reaction,
%  and in this spirit we use the term ``burning'' to refer to any
%  similar form of front propagation.}
%
(i.e. front propagation speed in the local fluid frame) is homogenous
and isotropic.

A front is the oriented boundary of a burned region with orientation
defined by the normal vector $\nhat$ pointing away from the burned
region. (We can also refer to the orientation using the tangent vector
$\ghat$ where $\nhat \times \ghat = +1$, i.e. pointing out of the
plane.)  Denoting by $\mathbf{r}$ the $xy$-position of a front element
and by $\theta$ the angle from the $x$-axis to $\ghat$, a front is a
curve in $xy\theta$-space that satisfies the \emph{front-compatibility
  criterion},
\begin{align}
\label{eq:front_compatibility_criterion}
\frac{d \rvec}{d \lambda} \propto \ghat(\theta),
\end{align}
where $\lambda$ is some smooth parameterization of the curve.  The
above assumptions lead to the following three-dimensional ODE
governing the evolution of an individual front element $(\mathbf{r}(t),\theta(t))$.
\begin{subequations}
\begin{align}
\dot{\mathbf{r}} &= \mathbf{u} + v_0 \hat{\mathbf{n}}, \label{eq:3DODEa} \\ 
\dot{\theta} &= - \hat{n}_i u_{i,j} \hat{g}_j, \label{eq:3DODEb}
\end{align}
\label{eq:3DODE}
\end{subequations}
where $\mathbf{u}$ is the prescribed fluid velocity field, which is
nondimensionalized by dividing by $U$, the maximum fluid vortex speed
in the absence of wind.  That is, in the absence of wind, the maximum
value of $u$ is unity.  Here, $v_0 = V_0/U$ is the nondimensionalized
front propagation speed in the comoving fluid frame.  The position
variable $\mathbf{r}$ is scaled so that the width of each vortex and
of the channel is unity.  Time is scaled by the advection time $D/U$,
where $D$ is the (dimensionful) vortex width.  Note that
$\hat{\mathbf{g}} = (\cos \theta, \sin \theta)$ and $\hat{\mathbf{n}}
= (\sin \theta, -\cos \theta)$ indicate the tangent to the front
element and the normal direction (propagation direction),
respectively. Furthermore, $u_{i,j} = \partial u_i / \partial r_j$ and
repeated indices are summed.  The total translational motion of a
front element is the vector sum of the fluid velocity and the front
propagation velocity in the fluid frame, Eq.~(\ref{eq:3DODEa}).  The
change in orientation is determined entirely kinematically;
Eq.~(\ref{eq:3DODEb}) describes the angular velocity of a material
line embedded in the fluid.

\begin{figure}[bt]
\includegraphics[width=1\linewidth]{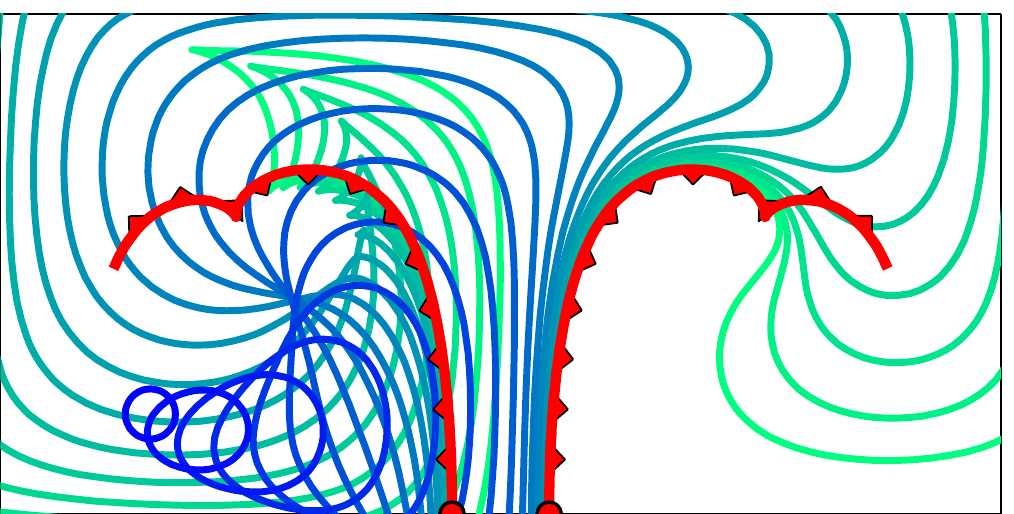}
\caption{(Color online.)  Evolution of reaction front (blue to green)
  in two counter-rotating vortices. Stimulation on lower left grows
  while being acted on by the flow. Two BIMs (red) emanate from BFPs
  on the bottom channel wall. The ``burning direction'' of each BIM is
  indicated by red triangles. The reaction passes through oppositely
  oriented BIM, but is blocked by cooriented BIM. Finally the reaction
  front wraps around cusp of right BIM.}
\label{fig:BIM_barrier}
\end{figure}

Invariant manifolds of the full 3D ($xy\theta$) dynamics,
Eq.~(\ref{eq:3DODE}), depend upon both the fluid flow and front
propagation, and therefore differ from the invariant manifolds of the
underlying advection dynamics.  We focus on the 1D unstable manifolds
attached to the burning fixed points (BFPs)---i.e. fixed points of
Eq.~(\ref{eq:3DODE})---that are of stability type stable-stable-unstable
(SSU).  We call these \emph{burning invariant manifolds} (BIMs).  It
has been demonstrated theoretically and experimentally that these BIMs
are ``one-way'' barriers to front propagation in flows
(Fig.~\ref{fig:BIM_barrier}).  That is, they prevent reactions from
crossing in one direction but allow them to cross in the other.  It is
somewhat surprising that these codimension-two manifolds are in fact
barriers.  BIMs are not generic curves through $xy\theta$-space; they
obey the front compatibility criterion
\ref{eq:front_compatibility_criterion} \cite{Mitchell12b}.  All
fronts, including BIMs, obey the front no-passing lemma: no front can
overtake another front from behind.

An interesting consequence of the front propagation dynamics is the
ability to create cusps in fronts and in the BIMs.  In
time-independent flows, cusps mark a change in the bounding nature of
BIMs.  Figure.~\ref{fig:BIM_barrier} illustrates the evolution of a
small circular front (lower left, blue).  During its evolution (blue
to green), it passes through the left BIM (red) because of their
opposite orientation.  It then presses up against the right BIM
(cooriented) and follows closely until reaching the BIM cusp where the
BIM's relative orientation changes, thus allowing passage of the
reaction front.  We define the \emph{BIM core} as the BIM segment that
includes the BFP and extends in both directions until reaching either
a cusp, a new BFP, or infinity.

%\begin{figure}[bt]
%\includegraphics[width=\linewidth]{../figures_and_code/BIM_theory_paper/combine_fronts_near_cusp-crop}
%\caption{\JM{this needs to show a BIM with a cusp and a front that wraps around. also mark the BIM core}}
%\label{fig:cusp_and_BIM core}
%\end{figure}

%%%%%%%%%%%%%%%%%%%%%%%%%%%%%%%%%%%%%%%%%%%%
\section{Frozen fronts: basic theory}
\label{sec:FF}

Consider a fluid domain $D$ that is connected, but not necessarily
simply connected.  In this paper, we focus on a channel flow where
$D = \mathbb{R} \otimes [0,1]$, but the results obtained in this
section are general.  We now introduce a more precise mathematical
definition of frozen front than the more intuitive definition used
thus far.  First, we define \emph{frozen domain} as a burned subdomain
of $D$ that is invariant under the burning dynamics and stable to
perturbation~\cite{Note3}.
%
%\footnote{These frozen domains are similar to the minimal
%  forward invariant sets studied in the context of random differential
%  equations with bounded noise \cite{Homburg10, Lamb14}.}.  
%
(See App.~\ref{app:stability} for a precise discussion of this notion
of stability).  Since the fluid is incompressible, neither the frozen
domain nor its complement may be of finite area.  A \emph{frozen
  front} (FF) is the oriented boundary of a frozen domain that
separates the burned from the unburned fluid.  (The frozen domain
boundary that coincides with the boundary of $D$, i.e. a domain wall,
is then not considered part of the frozen front.)  As with any front,
we choose the orientation of the FF to be a unit vector normal to the
FF pointing outward from the burned region.  Since the frozen domain
is unbounded, the FF cannot be a closed curve.

Consider a particular FF $F$ as a curve in $xy\theta$-space.  An
individual front element on $F$ can evolve into the interior of the
frozen domain, but not vice versa
(Fig.~\ref{fig:front_element_lifecycle}).  Since the frozen domain is
invariant, the time evolution of $F$ under Eq.~(\ref{eq:3DODE})
includes $F$ for any time t.  In other words, the backward trajectory
of any point on $F$ remains on $F$.  Thus the FF must be the union of
segments of front element trajectories, and is hence a piece-wise
smooth curve.  Each segment follows a trajectory from $t = -\infty$ to
some $t = t_f$.  This implies each segment lies within the unstable
manifold emanating from a fixed point, which may be at infinity (see
App.~\ref{app:far_field}).

\begin{figure}[bt]
\includegraphics[width=1\linewidth]{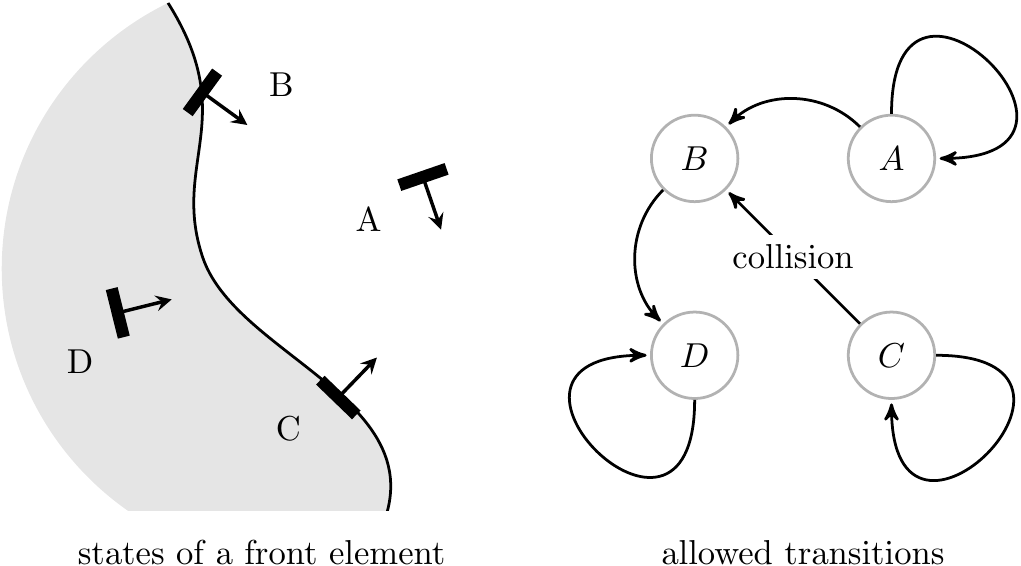}
\caption{An arbitrary front element exists in one of four states with
  respect to a burned region: $A$, unburned region; $B$, on the
  boundary of the burned region with non-outward-normal burning
  direction; $C$, on the boundary and oriented in the outward normal
  direction; $D$, inside the burned region.  The diagram on the right
  indicates how the state of a front element may change as it
  coevolves with the burned region.  These same dynamics hold between
  a front element and the fluid domain boundary.}
\label{fig:front_element_lifecycle}
\end{figure}

%How do you define the front when the region boundary has no derivatives?
%
%Front elements can go from physical to nonphysical, but not vice versa (Fig.~\ref{fig:front_element_lifecycle}).
%And this is a flow which gives us continuity / locality.
%This means that any element on the PF must flow from a nearby element on PF.
%
%We can then restrict the 3D ODE to a flow on the PF.
%\JM{Does this actually sneak in the PWS assumption already? The ``lifecycle'' figure might assume PWS.}
%
%Assume the pinning front PF is comprised of a set of disconnected PWS components.
%We can come back to the PWS assumption later.

On a smooth segment of FF a front element is either a fixed point of
the flow, or it ``slides'' along the segment satisfying
$\rdot \propto \ghat$.  Any FF can thus be decomposed into a
collection of these \emph{sliding fronts} (App.~\ref{app:sliding_fronts}).
%For a generic smooth flow, the fixed points will be isolated.
%Then each smooth segment is broken into subsegments, each of which is a portion of a sliding front.
%The restricted flow is also onto for each of these subsegments.
%Each subsegment ends either at a fp, $\partial D$, $\infty$, or another sliding front.
Here we summarize the geometry of sliding fronts detailed in
App.~\ref{app:sliding_fronts}.  First, sliding fronts only exist in
the domain where $|\uvec| \ge v_0$.  We refer to this domain as the
fast zone FZ, and the complementary domain as the slow zone SZ.  In
the FZ, the structure of the sliding fronts can be simply
characterized.  At every point in the FZ interior, there are two
allowed sliding front orientations characterized by the angle,
\begin{align}
\beta = \arccos(-v_0/|\uvec|),
\end{align}
between the front propagation direction $\nhat$ and the fluid flow
$\uvec$ (Fig.~\ref{fig:generic_intersection_of_sliding_fronts},
Lemma~\ref{lemma:SFsCrossStreamlinesAngle}).  In the limit
$v_0/|\uvec| \to 0$, the two sliding fronts become parallel (burning
in opposite directions) and align with the streamlines, thus
recovering the advective case.  We refer to these two choices of
orientation as ``$+$'' and ``$-$'' corresponding to
$\sgn(\rdot \cdot \ghat)$.  Each choice of orientation defines a set
of sliding fronts whose projection foliates the FZ.
%(Note: Do not confuse this double foliation with the forward / reverse time double foliation sometimes used to analyze 2D flows/maps.)
When the sliding fronts are considered as curves in $xy\theta$-space,
they foliate a two-dimensional surface which is a double-branched
covering of the FZ.  (See
Figs.~\ref{fig:hyperbolic_wpm_surface_and_flow} and
\ref{fig:elliptic_wpm_surface_and_flow} for examples.)

\begin{figure}[bt]
\includegraphics[width=0.6\linewidth]{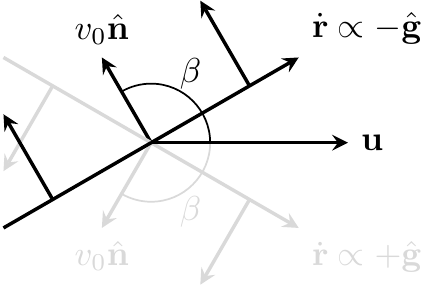}
\caption{A generic intersection of two sliding fronts (one black, one gray).  Each sliding front's propagation vector $v_0 \nhat$ cancels the normal component of the fluid velocity $\uvec$, leaving only motion tangent to the front.  The two orientations (black and gray) are symmetric about $\uvec$.}
\label{fig:generic_intersection_of_sliding_fronts}
\end{figure}

Consider a burned region bounded by two sliding fronts (on different
branches) that meet at a point as in
Fig.~\ref{fig:generic_intersection_of_sliding_fronts}.  In principle,
the burned region may be either locally concave or locally convex at
this point.  However, the convex case is not relevant to FFs because
any convex corner will be smoothed out after an arbitrarily short
evolution.  Therefore, in the interior of the FZ, a FF is simply a
union of smooth curves that meet at concave angles specified by the
local burning-to-fluid-speed ratio $v_0/|\uvec|$.  In the limit
$|\uvec| \to^+ v_0$, the two branches meet on the boundary of the SZ.
At all such points of the boundary, two sliding fronts meet with
burning directions $\nhat$ aligned.  There are two cases to consider.

In the first case, assume $\nhat$ is not perpendicular to the SZ.
Then the sliding front trajectory passes through the fold joining the two branches in such a way that it forms a cusp in the $xy$-plane (Fig.~\ref{fig:cusp_in_3D}).
\begin{figure}[bt]
\includegraphics[width=\linewidth]{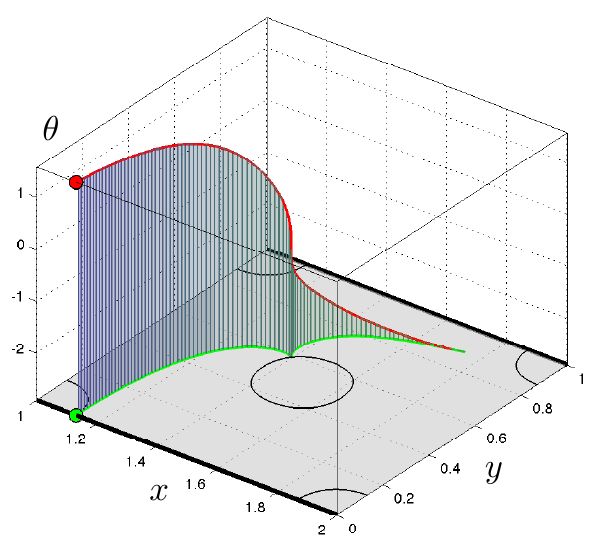}
\caption{(Color online.)  The BIM (red) is a smooth curve in $xy\theta$-space. Its
  projection (green) onto the $xy$-plane has a cusp on the boundary of
  the SZ.}
\label{fig:cusp_in_3D}
\end{figure}
We observed above that cusps mark a change in the bounding behavior of
BIMs.  This change occurs at cusps along any sliding front (including
BIMs), which implies that \emph{a FF cannot contain a cusp}.
Figure~\ref{fig:cusp_not_on_pinning_front} illustrates why; it shows
the two possible burned regions that would be bounded by such a cusp.
In both cases, one segment of the sliding front has a burning
direction incompatible with, i.e. pointing into, the proposed burned
region.
%Imagine a burned region bounded by a sliding front that has a cusp on some SZ (Fig.~\ref{fig:cusp_not_on_pinning_front}).
%Either the burned region would have a convex discontinuity at the intersection (disallowed by Lemma NO CONVEX PINNING REGION CORNER), or it would extend somewhat beyond the cusp and into the SZ.
%In the latter case, the entire SZ must eventually burn (Lemma BLAH).
%Then, since these two sliding fronts have the same burning direction and each foliation is smooth, the burned region would grow out of the SZ on the previously unburned side, making the proposed pinning front no longer a boundary between burned and unburned regions (Fig.~\ref{fig:cusp_not_on_pinning_front}).
%
\begin{figure}[bt]
\includegraphics[width=\linewidth]{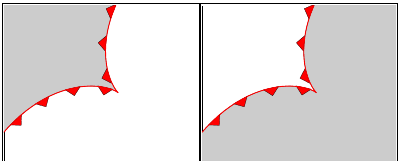}
\caption{(Color online.)  A sliding front (red) with a cusp cannot bound a burned
  region (gray). Either choice of shading leads to an incompatibility
  in front orientation in which one piece of the front points into the
  burned region.}
\label{fig:cusp_not_on_pinning_front}
\end{figure}

Referring to Fig.~\ref{fig:cusps_come_together}, as $\nhat$ becomes
perpendicular to the SZ at the point $\mathbf{x}$, the cusp becomes
tangent to the SZ.  By symmetry, a cusp also approaches $\mathbf{x}$
from the other side.

\begin{figure}[bt]
\includegraphics[width=1\linewidth]{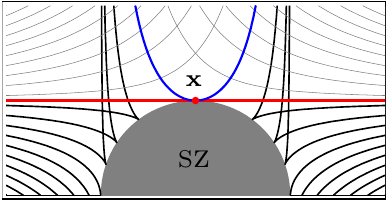}
\caption{(Color online.)  Sliding fronts (black) strike the SZ,
  forming cusps on either side of an SSU BFP $\mathbf{x}$.  As the
  cusps on either side approach $\mathbf{x}$, they become more
  horizontal, eventually joining tangent to each other at
  $\mathbf{x}$.}
\label{fig:cusps_come_together}
\end{figure}

In the second case, where $\nhat$ is perpendicular to the SZ,
Ref.~\cite{Mitchell12b} showed that the sliding front must meet the SZ
at a BFP $\mathbf{x}$.  This could be thought of as the meeting of two
cusps (Fig.~\ref{fig:cusps_come_together}).  Each segment of the cusp
on the left pairs with its symmetric segment on the right to form a
smooth curve in $xy$-space passing through $\mathbf{x}$.  Each of
these two combined curves is a 1D stable or unstable manifold of
$\mathbf{x}$.  There are four possible stability types of BFPs in
$xy\theta$-space: SSS, SSU, SUU, and UUU.  These are illustrated in
Fig.~\ref{fig:stability_table}.  For SSU and SUU BFPs, the dynamics
restricted to the sliding surface is of stability SU
(Lemma~\ref{cor:restricted_bfp_stability}).
Figure~\ref{fig:stability_table} illustrates the 1D stable and
unstable manifolds attached to such BFPs.  For SSS and UUU points, the
dynamics within the constraint surface is of stability SS and UU
respectively (Lemma~\ref{cor:restricted_bfp_stability}).  Since the
BFP is either a sink or source in this case, it is met by an infinite
number of sliding trajectories.

%\begin{figure}[bt]
%\includegraphics[width=\linewidth]{../figures_and_code/temporary/four_bfp_stability_types.jpg}
%\caption{SSU and SSS bfp configurations are stable with respect to burned region perturbation. SUU and UUU are not.}
%\label{fig:four_bfp_stability_types}
%\end{figure}

Only two of the four stability types can occur on a frozen front.
Suppose a frozen front is tangent to a SZ at a BFP where the burning
direction is into the SZ, as for SUU or UUU stability types.  Though
the burned region behind the BFP does not intersect the SZ, a small
perturbation of the burned region at the BFP can intersect the SZ.
Once any of the SZ is burned, the entirety of the SZ must eventually
be burned and remain burned forever
(Lemma~\ref{lem:any_SZ_is_all_SZ}).  Since we require frozen fronts to
be stable under small perturbations (App.~\ref{app:stability}), SUU and
UUU BFPs cannot occur on a FF.

The two remaining stability types SSU and SSS can exist on a FF.  We
previously showed that the FF consists of unstable manifolds.  Only
the SSU points have unstable manifolds.  Finally, since cusps are not
allowed on FFs (shown earlier), we have one of the main results of
this paper.

\begin{proposition}
\label{prop:FF}
Frozen fronts are built from BIM cores.  More precisely, each frozen
front is generated by some set $\mathcal{S}_{FF}$ of SSU BFPs.  The
frozen front is obtained by tracing the unstable manifold from each
point in $\mathcal{S}_{FF}$ until one of three things occurs: it
intersects any other BIM core emanating from $\mathcal{S}_{FF}$; it
intersects any domain boundary; or it terminates at an SSS BFP.
\end{proposition}

So far we have focused our attention on the generation of the FF from BFPs.
Here we shift our attention to consider how the sliding segments of a FF end.
We have already discussed the most common case where segments intersect at a concave corner.
There exist two other possibilities, termination on an SSU or SSS BFP.

%An SSU BFP also has an incoming sliding front within its stable manifold.
An SSU BFP has a stable manifold that contains an \emph{incoming}
sliding front.  A FF can therefore contain a segment which is a
heteroclinic connection consisting of a sliding front between SSU
points.  Figure~\ref{fig:SSU_to_SSU_connection}b shows two SSU BFPs
joined by such a connection flowing from the upper to the lower BFP.
This configuration is a FF; in particular, it is stable to
perturbations of the burned region (App.~\ref{app:stability}).  In one
sense, the FF is also structurally stable because generic
perturbations of the flow yield frozen domains with a similar shape
(Figs.~\ref{fig:SSU_to_SSU_connection}a and
\ref{fig:SSU_to_SSU_connection}c).  In another sense, it is not
structurally stable, because generic perturbations break the
heteroclinic connection, thus altering the dynamics along the front.
Some of these perturbations cause the lower SSU BFP to fall behind the
FF (Fig.~\ref{fig:SSU_to_SSU_connection}a), while other perturbations
cause it to push through, and in doing so contribute a segment of
unstable sliding front to the FF
(Fig.~\ref{fig:SSU_to_SSU_connection}c).
%In another sense, it is not structurally stable because perturbations break the heteroclinic connection and remove the ``target'' SSU BFP from the FF, while the ``source'' SSU points remain.
As seen in Figs.~\ref{fig:SSU_to_SSU_connection}b and
\ref{fig:SSU_to_SSU_connection}c, both of these perturbations return
the system to the generic case.  So while SSU BFPs can exist as
``termination points'' along a FF, this is not generic.

\begin{figure}[bt]
\includegraphics[width=\linewidth]{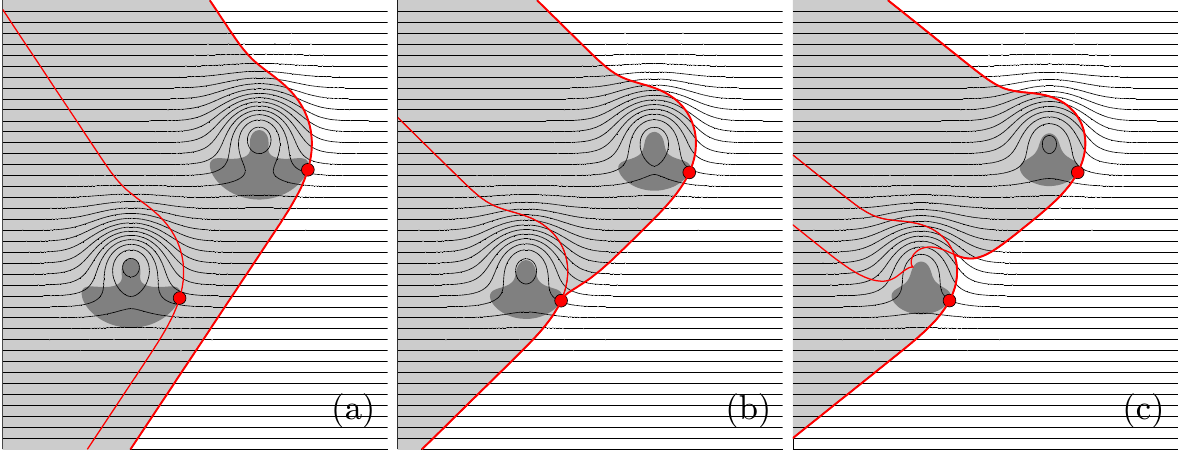}
\caption{(Color online.)  The SSU--SSU connection is not structurally
  stable as the wind speed is varied.  Nevertheless, the frozen domain
  (light gray) varies continuously. SZs are dark gray.  (a) The
  relation $v_w < v_{c}$ places the lower SSU point behind the FF attached
  to the upper SSU point.  (b) The equality $v_w = v_{c}$ makes the unstable
  manifold from the upper SSU point coincide with the stable manifold
  of the lower point.  (c) The relation $v_w > v_{c}$ pushes the lower
  SSU point ahead, placing it and its BIM on the FF. The FF is now composed
  of two BIMs meeting at a  concave corner.  }
\label{fig:SSU_to_SSU_connection}
\end{figure}

Finally, we consider the SSU to SSS connection.  The SSS point
attracts all points within a 3D neighborhood and, therefore, it
attracts all sliding fronts within some neighborhood on the invariant
sliding surface.  It might then seem that this SSS point can be on a
FF containing any of these incoming sliding fronts.  However, the
sliding front must reach the SSS point without having formed a cusp.
This can only happen if the eigenvalues of the SSS point are real (see
App.~\ref{app:SSS}).  Such SSS points do exist, albeit for what
appears to be a small parameter range.

\section{Theory: windy alternating vortex chain flow}
\label{sec:theory}

We continue our discussion of FFs using a simple numerical model of the experimental fluid flow.

%%%%%%%%%%%%%%%
\subsection{Numerical model}

The stream function that describes the flow is
\begin{align}
\label{eq:WAVC}
\Psi = \frac{1}{\pi} \sin(\pi x) \sin(\pi y) - v_w y,
\end{align}
where $u_x = d \Psi/dy$ and $u_y = - d \Psi/ dx$.
This model has been used in several previous studies, on both fluid mixing and reacting flows, yielding reasonable agreement with experiment.
Our intent here is to illustrate the theory of frozen fronts for a particular fluid flow, and to reproduce basic features of the experimental flow in Sec.~\ref{sec:experiments}.

There is a weak three-dimensional component to the vortex flow due to
Ekman pumping that carries fluid toward the vortex centers at the
bottom of the fluid layer and up through the vortex cores
\cite{Solomon03}.  This effect is not included in the model.  Also,
while the model has free-slip boundary conditions, this is certainly
not true in the experiment.  Nevertheless, the simplified free-slip
model of Eq.~(\ref{eq:WAVC}) has been used successfully in modelling
several experiments on passive transport and front propagation in
vortex flows \cite{Paoletti05, Solomon88, Camassa91, Cencini03,
  Abel01}.  The last term in Eq.~(\ref{eq:WAVC}) numerically models
the fluid wind observed in the moving frame of the vortices.

\subsection{Dynamical systems analysis}

We begin by considering a flow where the wind is of insufficient
strength to produce a FF (Fig.~\ref{fig:sim_no_pinning}a).  The
streamlines indicate that this is essentially a vortex flow, but with
a sinuous, left-moving jet.  In Fig.~\ref{fig:sim_no_pinning}b a small
circular stimulation (purple) is made in the lower left.  This circle
evolves outward to the left and right while being deformed by the
flow.  The rightward moving front is able to make slow progress
``upwind''.  Notice though that it is blocked at the vortex boundary
near the bottom and middle, and must wind around through the top of
the channel.  In this way, the reaction continues winding rightward
through the channel indefinitely (Fig.~\ref{fig:sim_no_pinning}c).

Figure~\ref{fig:sim_no_pinning}d illustrates all SZs (gray), SSU BFPs
(red), and BIMs (red with arrows indicating the burning direction) in
this system.  Two of the SZs contain the elliptic advective fixed
points in the vortex centers.  The others contain hyperbolic advective
fixed points on the channel walls.  Note that the SZs are slightly
offset from a square lattice.  This is due to the small wind added.
We show only the SSU BFPs since, as we will show
(Prop.~\ref{prop:FF}), they are the generators of the unstable
manifolds which combine to form FFs.  Each BFP lies on a SZ boundary
and, because it is SSU, is oriented away from the SZ.  The BIMs
emanating from these BFPs are similarly oriented.  Each BIM spirals
into a vortex center where it forms a cusp on an elliptic SZ (i.e. a
SZ that contains an elliptic advective fixed point).  Only the
incoming portion of the cusp is pictured because, as shown earlier,
the FF cannot contain cusps, and so the remainder of the BIM will not
be relevant.

Figure~\ref{fig:sim_no_pinning}e summarizes the dynamical structures
relevant to the behavior observed in Fig.~\ref{fig:sim_no_pinning}b
and~\ref{fig:sim_no_pinning}c.  The BIM core shown is responsible for
blocking front propagation at the bottom and center of the channel in
Fig.~\ref{fig:sim_no_pinning}b.  The transverse stability of the BIM
leads to the front's rapid convergence upon it
(Fig.~\ref{fig:sim_no_pinning}b).  As the front evolves further
(Fig.~\ref{fig:sim_no_pinning}c), it reaches the cusp at the end of
the BIM core and winds around it.  The BIM core does not form a
complete span across the channel, and thus does not form a
\emph{global} barrier to the propagation of fronts.  This is the
situation seen in experimental images
Figs.~\ref{fig:onset_of_pinning_exp}a, b, c.
\begin{figure}[bt]
\includegraphics[width=\linewidth]{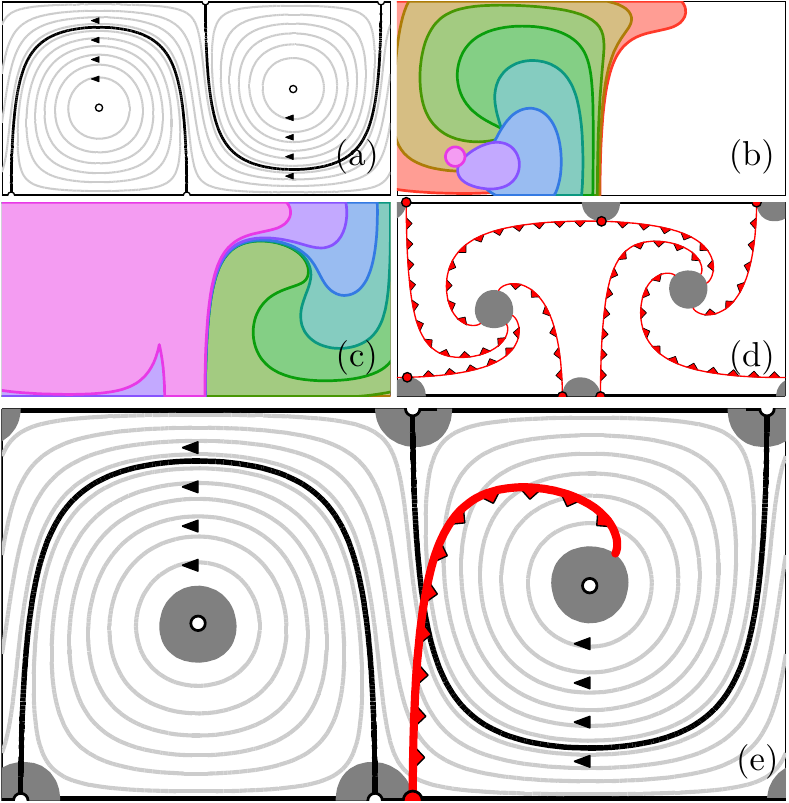}
\caption{(Color online.)  Small wind speed ($v_0 = 0.3, v_w =
  0.15$).
  (a) Fluid flow streamlines, fixed points and attached invariant
  manifolds. (b) Sequence of fronts shows preliminary convergence near
  bottom to roughly vertical curve. (c) Further evolution; lower edge
  converges to curved line while the rest proceeds around and to the
  right. (d) BIMs attached to BFPs. SZs shaded gray. (e) The one BIM
  most important for above front evolution - shown against advective
  structure.}
\label{fig:sim_no_pinning}
\end{figure}

Now we increase the wind speed until it precisely balances the burning
speed, $v_w = v_0$ (Fig.~\ref{fig:sim_critical_pinning}).  Stimulating
in the lower left (purple) we find that the reaction approaches a
vertical curve (Fig.~\ref{fig:sim_critical_pinning}b), and so the
reaction is confined to the left side.  This appears to be a candidate
for a frozen domain.  In Fig.~\ref{fig:sim_critical_pinning}c, we test
the stability of this region by introducing a small sinusoidal
perturbation.  The rightward component of this perturbation grows,
eventually filling in the entire cell to the right, demonstrating that
this region is not stable and therefore not a frozen domain.

Let us examine the dynamical structures in
Fig.~\ref{fig:sim_critical_pinning}d,e.  The increase in wind has
caused the SZs to shift slightly relative to
Fig.~\ref{fig:sim_no_pinning}d,e; the two on the lower boundary move
together, as do the two on the upper boundary; those in the vortices
move up or down depending on their rotational sense.  The central BIM
is now a straight vertical line.  It is important to note that this
BIM spans the entire channel with no cusp thus creating a
\emph{global} barrier to front propagation.  Symmetry of the flow
indicates that this BIM terminates at an SUU BFP at the top of the
channel.  We have previously argued that such a fixed point could not
be on a FF, and it is this fixed point that leads to the instability
demonstrated in Fig.~\ref{fig:sim_critical_pinning}c.

\begin{figure}[bt]
\includegraphics[width=\linewidth]{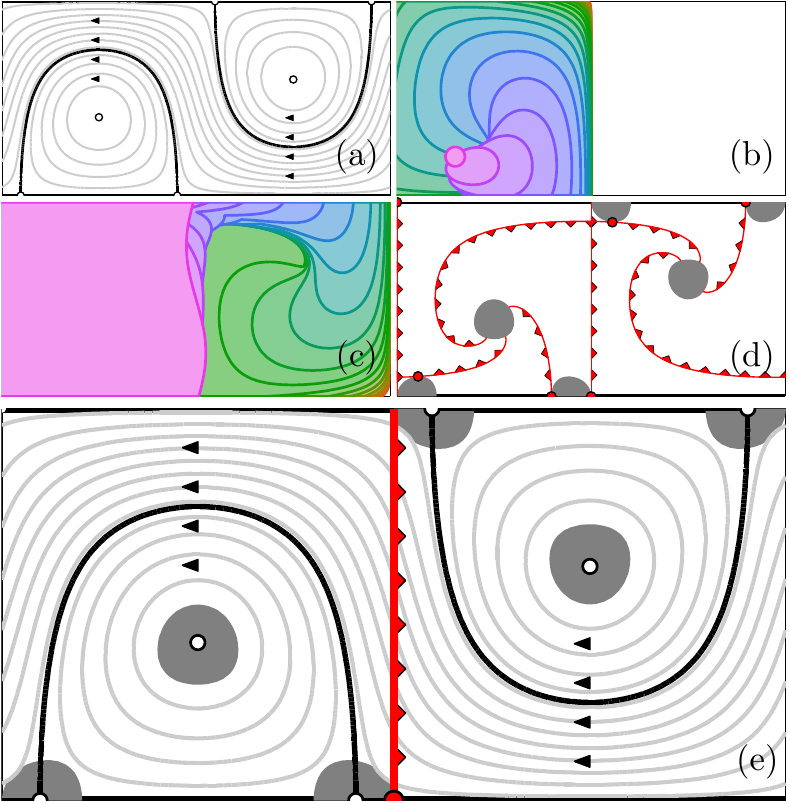}
\caption{(Color online.)  Critical wind speed ($v_0 = v_w = 0.3$). (a)
  Advective structure; similar to previous case. (b) This time, front
  progress (from the left) is completely blocked. (c) Perturbation of
  burned region shows instability. (d) Several BIMs, BFPs and SZs. (e)
  BIM of interest is a straight vertical line - coincides with
  separatrix of non-windy flow.}
\label{fig:sim_critical_pinning}
\end{figure}

Now we increase the wind beyond the critical value.  In
Fig.~\ref{fig:sim_type_1_pinning}a a stimulation on the left converges
to a burned region bounded by a smooth curve spanning the channel.
Unlike in Fig.~\ref{fig:sim_critical_pinning}c, a small rightward
perturbation in Fig.~\ref{fig:sim_type_1_pinning}b converges back to
this smooth curve, and hence this curve is a FF.
Figure~\ref{fig:sim_type_1_pinning}c shows that the smooth bounding
curve is the BIM emanating from the bottom BFP.  Note that this BIM
terminates at a point on the boundary that is not a BFP.  This
explains the situation seen in experimental images
Fig.~\ref{fig:onset_of_pinning_exp}d, e, f as well as
Fig.~\ref{fig:pinfronts3}a.

\begin{figure}[bt]
\includegraphics[width=\linewidth]{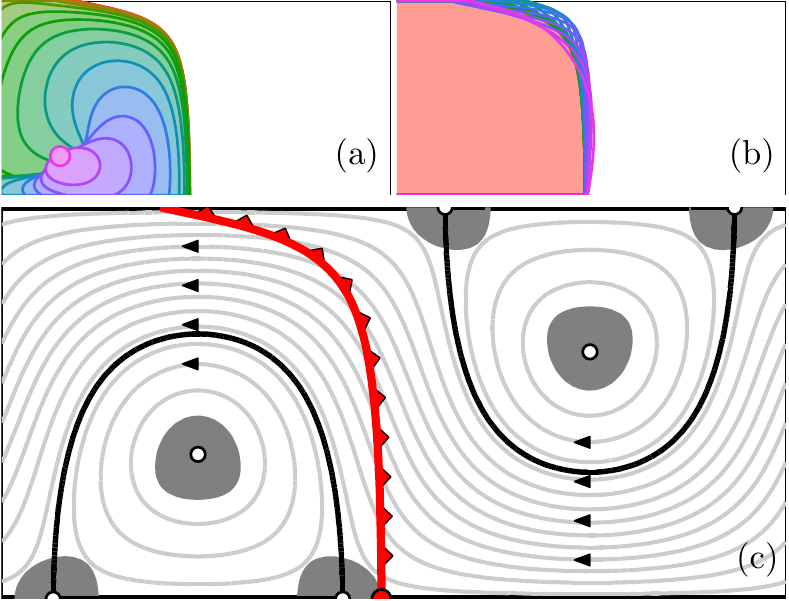}
\caption{(Color online.)  Wind greater than critical
  ($v_0 = 0.3, v_w = 0.4$). A stimulation on the left (a) converges
  onto a smooth curve that spans the channel.  In (b) a sinusoidal
  perturbation of this curve converges back to the curve, implying
  that it is stable.  (Only the last front is filled). (c) The BIM
  responsible for the FF spans the channel with no cusps.}
\label{fig:sim_type_1_pinning}
\end{figure}

Now that we have seen BIMs act as both local and global barriers, we
would like to understand the transition between these two cases in
more detail.  Imagine a deformation that takes the BIM in
Fig.~\ref{fig:sim_no_pinning}e to the BIM in
Fig.~\ref{fig:sim_type_1_pinning}c; What might this deformation look
like?  Lemma~\ref{lemma:SFsCrossStreamlinesAngle} ensures that the
angle between BIMs and streamlines is nonzero throughout the interior
of the FZ.  Therefore a BIM cannot form a tangency with the channel
wall (which must coincide with a streamline) in the interior of the
FZ.  Note, however, that a BIM \emph{cusp}, on the boundary of a SZ,
may encounter the channel wall without forming a tangency.  In fact,
this occurs when the cusp is perpendicular to the channel wall
(Lemma~\ref{lemma:SFOnFluidBoundary}).  This observation suggests two
deformation strategies: either move the existing cusp on the elliptic
SZ to the wall, or create a new cusp on the hyperbolic SZ and slide
the cusp to the wall.  While the first mechanism seems more
straightforward, and has not been ruled out theoretically, it has not
yet been observed.  However, the second mechanism is observed here.

In Fig.~\ref{fig:sim_cusp_formation} we increase $v_w$ through the critical value $v_w = v_0$ and follow the transformation of the BIM.
Beginning with a subcritical $v_w$ value in Fig.~\ref{fig:sim_cusp_formation}a, we see the BIM (green) that comes up from BFP $A$ (not shown) on the bottom wall and veers off to the right to form a cusp on the elliptic SZ (not shown).
This cusp marks the end of the BIM core.

Increasing the wind, the BIM is ``blown backward'' developing a
tangency (red and blue dashed) with the upper SZ.  This tangency is
not forbidden, because the SZ is not defined by a streamline.  Since
the front is burning away from the SZ, the tangency must occur at
either an SSU or SSS BFP on the upper SZ (according to
Lemma~\ref{lemma:SFIntersectsSZ} and Fig.~\ref{fig:stability_table}.)
Because the SZ is convex in this case, the BFP must be SSU.  The
heteroclinic connection is illustrated by the coincidence of the
unstable BIM of BFP $A$ and the stable BIM of BFP $B$ (red and blue
dashed).

Continuing to increase the wind, the BIM is blown further backward.
Now it does not meet the SZ at a tangency, and so the heteroclinic
connection is broken, giving way to a cusp, the other option allowed
by Lemma~\ref{lemma:SFIntersectsSZ}.  This cusp slides along the SZ,
with its angle changing to remain perpendicular to the fluid flow
(Lemma~\ref{lemma:SFIntersectsSZ}).  The cusp must rotate
counterclockwise, at least initially, so that its tangent points into
the SZ, as rotating clockwise would require the BIM to enter the SZ.

The BIM soon arrives at another tangency with the SZ
(Fig.~\ref{fig:sim_cusp_formation}b).  Here, however, the BIM is
burning into, rather than away from, the SZ.  This tangency implies a
heteroclinic connection with the SUU BFP $C$.  (Again, the UUU case
cannot occur because the SZ is convex; see
Fig.~\ref{fig:stability_table}.)  In a generic three-dimensional
dynamical system, heteroclinic connections between SSU and SUU fixed
points are codimension-two occurrences.  In this sytem, however, the
BIMs are constrained to the two-dimensional sliding surface, and so
the heteroclinic connection is a codimension-one occurrence.  Said
loosely, if a BIM is to sweep from one side of a SZ to the other, the
BIM cannot avoid connecting with at least two BFPs on the SZ boundary.

\begin{figure}[bt]
\includegraphics[width=0.8\linewidth]{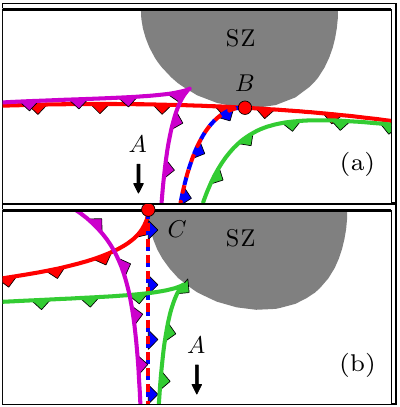}
\caption{(Color online.)  The basic mechanism in the transition to the
  first FF.  We increase the wind speed, showing the interplay between
  the BIM from BFP $A$ (not shown) and the upper SZ and its BFPs. (a)
  An SSU BFP $B$ lies on the bottom of the SZ.  Attached to it are
  BIMs (red) going left and right, both of which end in cusps on
  elliptic SZs (not shown).  (i) A BIM (green) comes up from the SSU
  BFP $A$ below (not shown) and then shadows the unstable BIM (red)
  going to the right.  (ii) The BIM (red and blue dashed) forms a
  tangency/heteroclinic connection with the BFP $B$.  (iii) The BIM
  (purple) is blown behind the heteroclinic connection, forming a
  cusp.  (b) An SUU BFP $C$ is shown at the top of the channel.  A BIM
  (red) lies within its unstable manifold and goes off to the left.
  (iv) A BIM (green) slides leftward along the SZ, approaching the BFP
  $C$.  (v) The BIM (red and blue dashed) forms a second
  tangency/heteroclinic connection with BFP $C$.  (vi) The BIM
  (purple) is blown beyond this heteroclinic connection, forming a
  complete span across the channel.  Since the BFP positions and SZs
  change slightly with $v_w$, the specific BFPs $B$ and $C$ shown, as
  well as their SZs, are calculated for the parameter values of the
  heteroclinic connections.}
\label{fig:sim_cusp_formation}
\end{figure}

Increasing the wind still further, the BIM, blown entirely clear of
the SZ, spans the entire channel, uninterrupted by cusps.  We have now
arrived at the FF configuration in Fig.~\ref{fig:sim_type_1_pinning}c.
This FF topology persists for a significant range of wind values.  As
seen in Fig.~\ref{fig:sim_type_1_pinning_multi}, the shape of this
front can be nearly straight, or more boomerang-shaped, depending on
the applied wind.  Note that it is only due to the symmetry of the
flow that the second heteroclinic connection in
Fig.~\ref{fig:sim_cusp_formation}b occurs exactly when the BIM core
first spans the channel

%\JM{note to self: the BIM cusp on the SZ has angle that changes as the cusp moves between BFPs. we know that cusps are always perp to the flow, not the SZ. we als know that the index theory of BFPs is based on the rotation of the flow along the SZ. this is eactly mirrored by the BIM cusp. what utility can we derive by considering the index theory applied to BIM cusps?}

%
\begin{figure}[bt]
\includegraphics[width=1\linewidth]{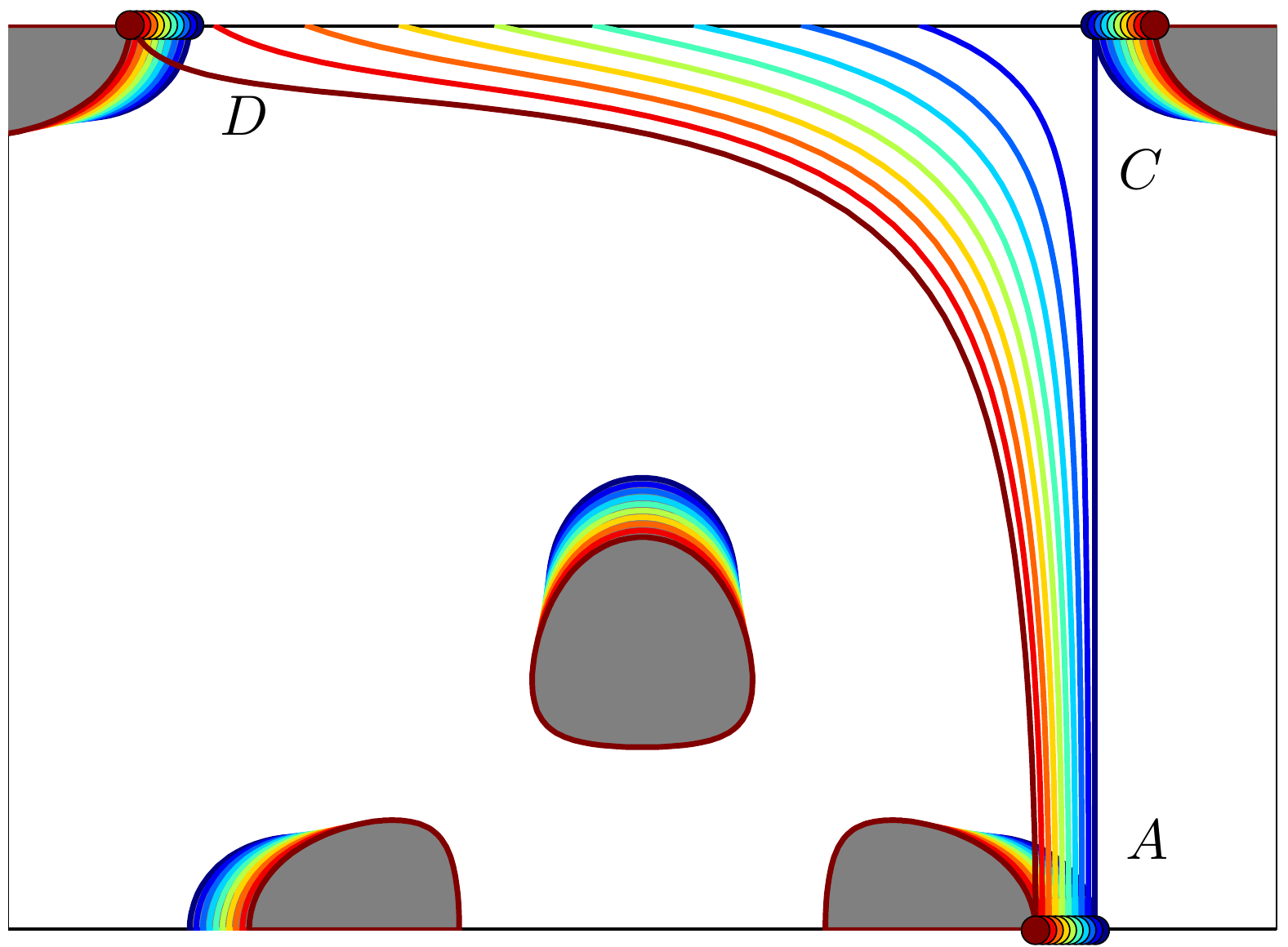}
\caption{(Color online.)  A series of FFs for increasing wind values
  ($v_0 = 0.3, v_0 < v_w < 1.7 v_0$, blue to red). The BIM attached to
  BFP $A$ is swept backward until it intersects BFP $D$. The SZs also
  shift and are colored accordingly.}
\label{fig:sim_type_1_pinning_multi}
\end{figure}

At approximately wind value $v_w = 0.34 = 1.7 v_0$, the BIM encounters
the upper left SZ in Fig.~\ref{fig:sim_type_1_pinning_multi}.  Just
like the BIM/SZ transition in Fig.~\ref{fig:sim_cusp_formation}, we
must form a tangency/heteroclinic connection
(Lemma~\ref{lemma:SFIntersectsSZ} and Fig.~\ref{fig:stability_table})
to a new SSU BFP denoted BFP $D$.  (Note the latter is rotated by
roughly $\pi/2$ CCW in comparison to
Fig.~\ref{fig:sim_cusp_formation}.)  Foretelling this tangency, the
red curves in Fig.~\ref{fig:sim_type_1_pinning_multi} begin to curve
upward near the upper channel wall.  Once again, symmetry of the flow
requires that BFP $D$ be on the upper channel wall.

As $v_0$ is increased still further, the BIM forms a cusp just behind
the unstable BIM attached to BFP $D$
(Fig.~\ref{fig:sim_type_2_pinning}b), as seen in the mechanism in
Fig.~\ref{fig:sim_cusp_formation}a.  Note that while a front may wrap
around the newly formed cusp attempting to bypass the initial BIM, it
will shortly encounter the BIM attached to BFP $D$ which has closed
off this pathway (Fig.~\ref{fig:sim_type_2_pinning}a).  Here we have a
FF that is composed of two distinct BIMs.  Note that the burning
region has a concave corner, with opening angle given by
Lemma~\ref{lemma:SFsCrossStreamlinesAngle}.  The appearance of this
concave corner is exactly what was observed in the experimental FF
(Fig.~\ref{fig:pinfronts3}b).

\begin{figure}[bt]
\includegraphics[width=1\linewidth]{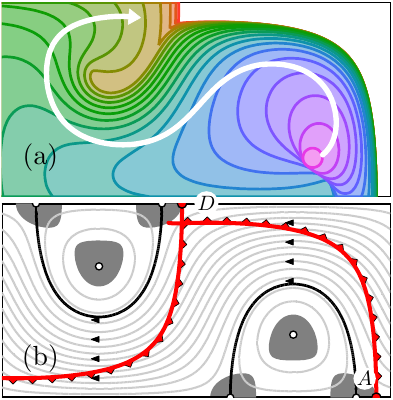}
\caption{(Color online.)  A composite FF formed from two BIMs.
  ($v_0 = 0.3, v_w = 0.525$).  The evolving front rapidly converges to
  BIM $A$ and then winds around its cusp. However, it is prevented
  from going further rightward by the short segment of BIM $D$.}
\label{fig:sim_type_2_pinning}
\end{figure}
\begin{figure}[bt]
\includegraphics[width=1\linewidth]{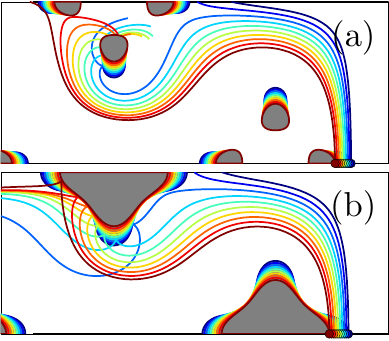}
\caption{(Color online.)  Increasing the wind beyond first instance of
  FF generates new transitions. (a) ($v_0 = 0.2, 0.3 < v_w < 0.6$)
  Blue FF rapidly attaches to and detaches from SZ. Upon detaching,
  the BIM ``jumps'' straight to a cusp on an elliptic SZ. With even
  higher wind, another attachment / detachment leads again to a
  complete span. (b) ($v_0 = 0.3, 0.45 < v_w < 0.8$) Illustration of
  similar transitions where ``jumping'' occurs all within a single
  connected SZ.}
\label{fig:sim_type_2_pinning_multi}
\end{figure}

The sequence in Fig.~\ref{fig:sim_type_2_pinning_multi}a takes the BIM
through a series of encounters with SZs as the wind speed is
increased.  (Here we consider $v_0 = 0.2$ for simplicity.)  The first
encounter is the attachment/detachment mechanism with the upper right
SZ, analogous to that in Fig.~\ref{fig:sim_cusp_formation}.  Here,
however, the BIM detachment does not result in a BIM core that spans
the channel.  Rather, the BIM continues for some distance and then
spirals in toward the elliptic SZ in the upper left, where it forms a
cusp.  As the wind speed increases, the cusp slides clockwise around
the SZ until the BIM forms a new tangency with the upper left SZ.  The
cusp on the elliptic SZ is ``cut off'' by this tangency, which
dynamically precedes it along the BIM.  This begins the mechanism of
Fig.~\ref{fig:sim_cusp_formation} again, after which the BIM core
forms a complete span and defines a FF.

Figure~\ref{fig:sim_type_2_pinning_multi}b shows a similar sequence as
Fig.~\ref{fig:sim_type_2_pinning_multi}a for $v_0 = 0.3$.  The main
difference between these two images is that three SZs have merged into
one in Fig.~\ref{fig:sim_type_2_pinning_multi}b.  Consequently, the
initial detachment of the BIM from the upper right of the SZ results
in a new cusp formed near the bottom of the same SZ.  Furthermore, as
the cusp moves clockwise around the SZ, it is never ``cut off'', but
instead slides along the SZ to the channel wall.

By flip-shift symmetry of the flow, the BIM attached to BFP $D$ has
undergone the same transition as the BIM attached to BFP $A$ and so
forms a FF as well (Fig.~\ref{fig:sim_type_3_pinning}a).  Importantly
these two FFs intersect.  Consequently, in addition to the frozen
domains defined by single BIMs, the \emph{union} of two neighboring
frozen domains defines a distinct frozen domain.  This union is
continuously related to the frozen domain observed in
Fig.~\ref{fig:sim_type_2_pinning}a as $v_w$ is increased.  In
Figs.~\ref{fig:sim_type_1_pinning} and \ref{fig:sim_type_2_pinning},
there is a 1-to-1 correspondence between frozen domains and vortices
in the channel.  Now in Fig.~\ref{fig:sim_type_3_pinning}, the
diversity of frozen domains (at fixed $v_0$ and $v_w$) has increased.
We can have either a FF formed by a single BIM core
(Fig.~\ref{fig:sim_type_3_pinning}b), or by two intersecting BIM cores
(Fig.~\ref{fig:sim_type_3_pinning}c,d,e).  Note that the diversity of FFs
in Figs.~\ref{fig:sim_type_3_pinning}(b-e) is produced by small
changes to the initial stimulation point.

\begin{figure}[bt]
\includegraphics[width=1\linewidth]{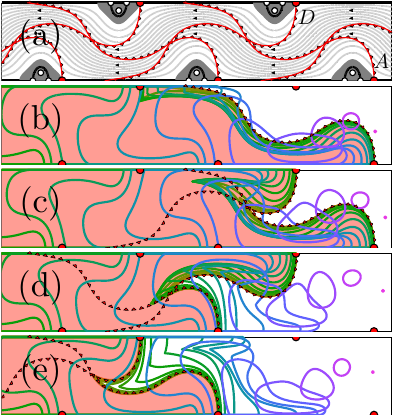}
\caption{(Color online.)  FF diversity and sensitivity to initial
  stimulation. ($v_0 = 0.3, v_w = 0.95$). (a) BIMs $A$ and $D$
  (related by flip-shift symmetry) each form a complete span, and
  intersect. Nearby stimulations (small pink dots near the right side)
  lead to different asymptotic frozen domains.  The frozen domains
  fall into two classes: (b) and (c-e). (b) The FF is composed of a
  single BIM, which spans the channel. (c-e) The FF is composed of two
  BIMs.}
\label{fig:sim_type_3_pinning}
\end{figure}

As the wind is increased, the process in
Fig.~\ref{fig:sim_type_2_pinning_multi} is repeated.  The BIM slides
along the upper channel wall until it encounters an SSU BFP on a SZ.
It moves around the SZ until it moves completely to the left of the SZ
and reconnects to the channel wall.  This process occurs once for each
vortex pair.  After each such occurrence, the BIM acquires a new
intersection with another BIM emanating from the opposite side of the
channel.  We can thereby enumerate all frozen domains of this system
for given values of $v_0$ and $v_w$.
(Fig.~\ref{fig:WAVC_front_enumeration}).  Finally, when
$\min(|\uvec|) > v_0$, there are no SZs, and therefore no BFPs, and
therefore no FFs.

\begin{figure}[bt]
\includegraphics[width=\linewidth]{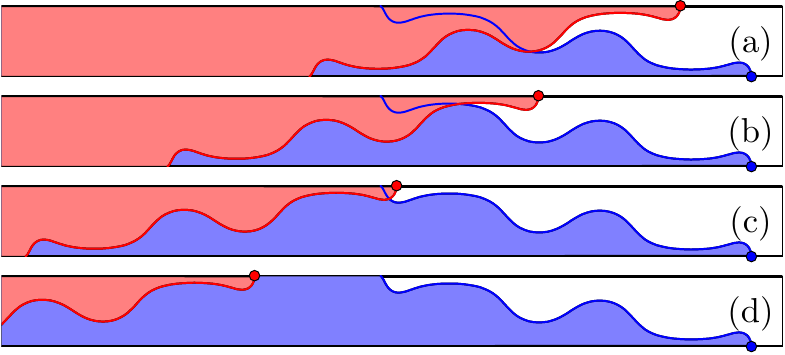}
\caption{(Color online.)  For the windy alternating vortex chain flow
  we can enumerate the increasing number of possible frozen domains
  that occur with increasing wind speed. In this example, there are
  four FF shapes (up to flip-shift symmetry).
  ($v_0 = 0.2, v_w = 1.18$)}
\label{fig:WAVC_front_enumeration}
\end{figure}
%

%%%%%%%%%%%%%%%%%%%%%%%%%%%%%%%%%%%%%%%%%%%%
\section{Conclusions}
\label{sec:conclusion}

The ability of a heterogeneous flow to freeze reaction fronts in the
presence of an imposed wind appears to be quite general.  Frozen
fronts (``sustained patterns'') have been observed numerically in
simulations of oceanic plankton blooms \cite{plankton2004}.  Frozen
fronts have also been seen both experimentally and numerically in
reacting flows in a porous media with a through-flow
\cite{kaern02,porouspinning13}.  We have also conducted experiments on
frozen fronts in extended flows composed of two-dimensional arrays of
vortices, either ordered or disordered~\cite{Megson15}.  As is the
case in this paper, the frozen fronts in an extended flow with a wind
are due to patterns of overlapping BIMs.

This work suggests several directions of future research.
In the context of design and control, this analysis could be used to develop a reacting fluid flow with some desired property.
An obvious example is a system with maximal reaction rate.
Given some class of accessible fluid flows, the reaction rate can be readily maximized by computing the lengths of FFs.
Another example is reaction rate stability. 
We might be given a particular flow perturbation and seek the base flow that minimizes reaction rate fluctuation.

It may be desirable to generate a FF with a particular geometry.
For instance, there may be a region in the neighborhood of the FF that we wish to keep strictly separated from the front
(e.g. a sensor in the vicinity of a combustion front that cannot withstand the temperatures of the front itself).
The analysis here provides a detailed connection between the stream function and FF shape making these questions accessible.

%%%%%%%%%%%%%%%%%%%%%%%%%%%%%%%%%%%%%%%%%%%%
\acknowledgements The present work was supported by the US National
Science Foundation under grants PHY-0748828 and CMMI-1201236
(Mitchell) and grants DMR-1004744, DMR-1361881 and PHY-1156964
(Solomon).

\appendix

%%%%%%%%%%%%%%%%%%%%%%%%%%%%%%%%%%%%%%%%%%%%
\section{Sliding fronts}
\label{app:sliding_fronts}

Although we study fronts propagating in time-independent fluid flows,
the fronts themselves certainly need not be time-invariant.  For
instance, a fast-propagating front in a weak flow will evolve
approximately as a circle of increasing radius.  Loosely speaking,
this is because each front element in the circle ``burns beyond
itself''.  For a front to be time-invariant, each element must instead
``slide along itself''.  In this section, we make this statement clear
and derive several consequences.
\begin{definition}
A front element, i.e. a point in $xy\theta$-space, is said to be \emph{sliding} when $\rdot \propto \ghat$, where $\ghat = [\cos\theta, \sin\theta]$. Equivalently, 
\begin{align}
\label{eq:sliding_constraint}
\rdot \cdot \nhat = 0,
\end{align}
where $\nhat = [\sin\theta, -\cos\theta]$.
\end{definition}
While the sliding property is defined for any fluid flow, it is of most use when the flow is time-independent, as we have assumed throughout this paper and its appendices.

The ``sliding'' constraint Eq.~(\ref{eq:sliding_constraint}) is
illustrated geometrically in
Fig.~\ref{fig:generic_intersection_of_sliding_fronts}.  For a given
$xy$ location, there are either zero, one or two solutions for
$\theta$ satisfying this constraint.  Where the fluid speed is small
($|\uvec| < v_0$), there is no solution; we call such a region a ``slow
zone'' (SZ). 
\begin{lemma}
\label{lemma:NoSFinSZ}
There are no sliding elements inside a SZ.
\end{lemma}
\begin{proof}
  Combining the sliding constraint Eq.~(\ref{eq:sliding_constraint})
  with Eq.~(\ref{eq:3DODE}), we find $|\uvec \cdot \nhat| = |v_0|$.
  This cannot be satisfied for $|\uvec| < v_0$.
\end{proof}
Where the fluid speed is large ($|\uvec| > v_0$), there are two
solutions to Eq.~(\ref{eq:sliding_constraint}).  Where the fluid speed
equals the burning speed ($|\uvec| = v_0$), these two solutions are
degenerate.  We call a region where $|\uvec| \ge v_0$ a ``fast zone''
(FZ).  The sliding constraint Eq.~(\ref{eq:sliding_constraint})
defines a two-dimensional submanifold of $xy\theta$-space, called the
\emph{sliding surface}, which can be viewed as a double-branched
surface over the FZ.
Figures~\ref{fig:hyperbolic_wpm_surface_and_flow}a and
\ref{fig:elliptic_wpm_surface_and_flow}a show the sliding surface for
a hyperbolic and an elliptic flow, respectively~\cite{Note4}.
%
%\footnote{Interestingly, these two sliding surfaces are related by
%  $x \leftrightarrow y$ symmetry.}.  
%
In Figs.~\ref{fig:hyperbolic_wpm_surface_and_flow}a and
\ref{fig:elliptic_wpm_surface_and_flow}a we see that, when viewed from
above, these sliding surfaces have a hole in the middle exactly where
the SZ is.
\begin{figure}[bt]
\includegraphics[width=1.0\linewidth]{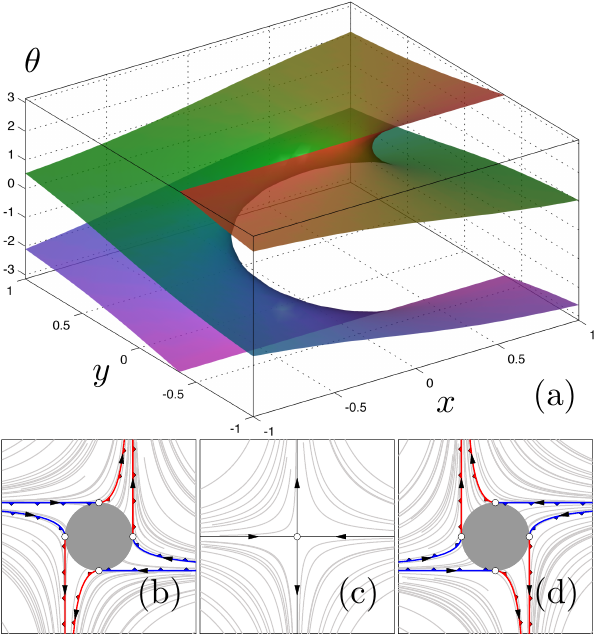}
\caption{(Color online.)  Hyperbolic fluid flow.
  $\dot{\mathbf{x}} = -A \mathbf{x}, \dot{\mathbf{y}} = +A
  \mathbf{y}$.
  ($v_0 = 0.35$, $A = 1$) (a) Sliding
  surface. (b) Streamlines of the $\mathbf{w}_+$ field. (c)
  Streamlines of the advective fluid flow. (d) Streamlines of the
  $\mathbf{w}_-$ field.}
\label{fig:hyperbolic_wpm_surface_and_flow}
\end{figure}
\begin{figure}[bt]
\includegraphics[width=1.0\linewidth]{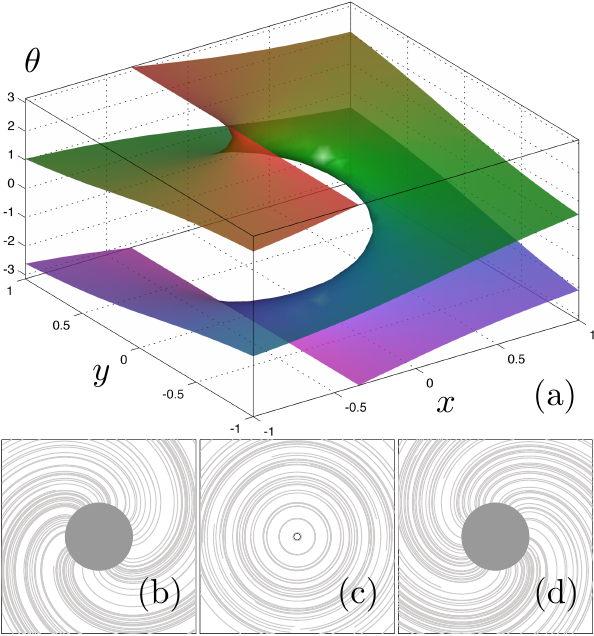}
\caption{(Color online.)  Elliptic fluid flow.
  $\dot{\mathbf{x}} = -A \mathbf{y}, \dot{\mathbf{y}} = +A
  \mathbf{x}$.
  ($v_0 = 0.35$, $A = 1$) (a) Sliding surface. (b) Streamlines of the
  $\mathbf{w}_+$ field. (c) Streamlines of the advective fluid
  flow. (d) Streamlines of the $\mathbf{w}_-$ field.}
\label{fig:elliptic_wpm_surface_and_flow}
\end{figure}

\begin{lemma}
\label{lem:any_SZ_is_all_SZ}
If at any time some portion of a SZ is burned, the asymptotic burned
domain will include that entire SZ.
\end{lemma}
\begin{proof}
Within the SZ, the velocity of the front is everywhere greater than the fluid.
Therefore, no direction of motion is forbidden to the front, and so the front will eventually access all parts of the SZ.
\end{proof}

\begin{lemma}
\label{lemma:sliding_invariant}
Sliding is an invariant property. That is, if a front element is sliding, every element along its trajectory under Eq.~(\ref{eq:3DODE}) is also sliding. Hence we may speak of \emph{sliding trajectories}.
\end{lemma}
\begin{proof}
We examine the time derivative of Eq.~(\ref{eq:sliding_constraint}).
%%
%\begin{align*}
%\frac{d}{dt} (\rdot \cdot \nhat) &= \hat{n}_i u_{i,j} \dot{r}_j + \dot{r}_k \frac{d \hat{n}_k}{dt}\\
%&= \hat{n}_i u_{i,j} \dot{r}_j + \dot{r}_k \hat{g}_k \dot{\theta}\\
% &= |\dot{\mathbf{r}}| \hat{n}_i u_{i,j} \hat{g}_j + |\dot{r}| \hat{g}_k \hat{g}_k \dot{\theta}\\
% &= |\dot{\mathbf{r}}| (-\dot{\theta}) + |\dot{\mathbf{r}}| \dot{\theta} = 0
%\end{align*}
%%
%
\begin{align*}
  \frac{d}{dt} (\rdot \cdot \nhat) 
&= (u_{i,j} \dot{r}_j + v_0  \dot{\theta} \hat{g}_i) \hat{n}_i  + \dot{r}_k \hat{g}_k  \dot{\theta} \\
&= \hat{n}_i u_{i,j} \dot{r}_j +  \dot{r}_k  \hat{g}_k  \dot{\theta}  \\
&= \pm  |\dot{\mathbf{r}}| \hat{n}_i u_{i,j} \hat{g}_j \pm |\dot{\mathbf{r}}|
  \hat{g}_k \hat{g}_k \dot{\theta}\\  
&= \pm |\dot{\mathbf{r}}| (-\dot{\theta}) \pm |\dot{\mathbf{r}}| \dot{\theta} = 0,
\end{align*}
where the first equality follows from Eq.~(\ref{eq:3DODE}a) and the
fact that $d \hat{\mathbf{n}}/dt = \hat{\mathbf{g}} \dot{\theta}$, the
second from the orthogonality of $\nhat$ and $\ghat$, the third from
the sliding assumption $\rdot = \pm |\dot{\mathbf{r}}| \ghat$, and the
fourth from Eq.~(\ref{eq:3DODE}b).
\end{proof}
\noindent
A consequence of this lemma is that the sliding surface is dynamically
invariant.

Recall that a front is a curve
$(\mathbf{r}(\lambda), \theta(\lambda))$ that everywhere satisfies the
front compatibility
criterion~Eq.~(\ref{eq:front_compatibility_criterion}), which is
expressed equivalently as
\begin{align}
\label{eq:front_compatibility_criterion_nhat}
\drdlambda \cdot \nhat = 0.
\end{align}
\begin{lemma}
\label{lemma:sliding_trajectory}
A trajectory of Eq.~(\ref{eq:3DODE}) is sliding if and only if the
curve it sweeps out is a front.
\end{lemma}
\begin{proof}
  Choosing $\lambda = t$, Eq.~(\ref{eq:sliding_constraint}) is
  equivalent to Eq.~(\ref{eq:front_compatibility_criterion_nhat}).
\end{proof}
In light of Lemma~\ref{lemma:sliding_trajectory}, we may refer to a sliding trajectory as a \emph{sliding front}.
More generally, we make the following definition.
\begin{definition}
  A \emph{sliding front} is a smooth curve that everywhere satisfies
  Eq.~(\ref{eq:sliding_constraint}), or equivalently
  Eq.~(\ref{eq:front_compatibility_criterion_nhat}).
\end{definition}
Note that a sliding front may be composed of multiple trajectories
joined at fixed points.  Also, any segment of a sliding front is also
referred to as a sliding front.

\begin{lemma}
  BIMs are sliding fronts, and thus lie within the sliding surface.
\end{lemma}

\begin{proof}
  A BIM is the unstable invariant manifold of an SSU BFP. Since we
  consider time-independent flows, this invariant manifold is also a
  trajectory. As shown in Ref.~\cite{Mitchell12b}, BIMs satisfy the
  front compatibility criterion. Therefore, by
  Lemma~\ref{lemma:sliding_trajectory}, BIMs are sliding fronts.
\end{proof}

Since the sliding surface is invariant, it is natural to restrict
Eq.~(\ref{eq:3DODE}) to this surface.  We next derive an explicit
expression for this 2D flow.  Applying the sliding constraint
Eq.~(\ref{eq:sliding_constraint}) to the front element dynamics
Eq.~(\ref{eq:3DODE}), we have $\uvec \cdot \nhat = -v_0$. Using this
to resolve the unit vector $\nhat$ into components, we have
$\nhat \cdot \hat{\uvec} = -v_0 / |\uvec|$ and
$\nhat \cdot \hat{\uvec}^{\perp} = \mp \sqrt{1 - (v_0 / |\uvec|)^2}$,
where $\uvec^{\perp}$ or $\hat{\uvec}^{\perp}$ is a righthanded
rotation by $\pi/2$ of $\uvec$ or $\hat{\uvec} = \uvec / |\uvec|$.
Inserting the resolved form of $\nhat$ into Eq.~(\ref{eq:3DODE}) we
have,
\begin{align}
\label{eq:wpm_fields}
\rdot = \mathbf{w}_{\pm} \equiv \left[1 - \left( \frac{v_0}{|\uvec|} \right)^2 \right] \uvec \mp \frac{v_0}{|\uvec|} \sqrt{1 - \left( \frac{v_0}{|\uvec|} \right)^2} \uvecperp
\end{align}
This defines two flows over the FZ, one for each of the two branches
of the sliding surface.  Our sign convention is such that
$\mathbf{w}_+$ is the flow on the branch where
$\rdot = + |\rdot| \ghat$ (the $+$ branch), and $\mathbf{w}_-$ is the
flow on the branch where $\rdot = - |\rdot| \ghat$ (the $-$ branch).

Equation~(\ref{eq:wpm_fields}) shows that the $\mathbf{w}_{\pm}$
fields are undefined (complex-valued) within the SZ and are zero on
its boundary, confirming Lemma~\ref{lemma:NoSFinSZ}.  All fixed points
of Eq.~(\ref{eq:3DODE}) are fixed points of Eq.~(\ref{eq:wpm_fields})
because all BFPs trivially satisfy the sliding constraint.  These BFPs
lie on the SZ boundary.  Equation~(\ref{eq:wpm_fields}) also has a set
of spurious fixed points at all other points along the SZ boundary,
i.e. where $|\mathbf{u}| = v_0$.  However, we ignore these spurious
fixed points as they are not physically relevant fixed points of
Eq.~(\ref{eq:3DODE}), but rather result from the square-root
singularity in Eq.~(\ref{eq:wpm_fields}), obtained by projecting
Eq.~(\ref{eq:3DODE}) onto $xy$-space.  This square-root singularity
also invalidates the uniqueness of solutions to
Eq.~(\ref{eq:wpm_fields}) at the SZ boundary.  Thus, there are other
physically relevant trajectories that pass through the SZ boundary.

Figures~\ref{fig:hyperbolic_wpm_surface_and_flow}b,d and
\ref{fig:elliptic_wpm_surface_and_flow}b,d illustrate the
$\mathbf{w}_{\pm}$ flows for the cases of hyperbolic and elliptic
fluid flow, respectively.  The gray regions are the SZs.  BFPs are
indicated on the boundary of the hyperbolic SZ.  The stable and
unstable manifolds of these BFPs are shown in blue and red,
respectively.

Reference~\cite{Mitchell12b} proved that for any BFP, $\ghat$ is an
eigenvector of $u_{i,j}$, i.e.
\begin{align}
\label{eq:mu}
u_{i,j} \hat{g}_j = \mu \hat{g}_i,
\end{align}
where $\mu$ is the eigenvalue.  Reference~\cite{Mitchell12b} also
defined the quantity $\mu'$,
\begin{align}
\label{eq:muprime}
\mu' = \mu + v_0 \kappa,
\end{align}
where $\kappa$ is the signed curvature of the SZ boundary at the BFP.
($\kappa < 0$ means that $\nhat$ points toward the center of
curvature.)  Below, we reproduce Theorem~4 from
Ref.~\cite{Mitchell12b}, \cite{Note5}  
%
%\footnote{We correct a typo in the original: in the
%bottom row of the table, $\mu'>0$ now reads $\mu'<0$.}.

\begin{theorem}
\label{t4}
   For a time-independent, incompressible flow $\mathbf{u}$, the eigenvalues
   about a BFP are
\begin{align}
\lambda_0 &= -\mu, \label{r27}\\
\lambda_\pm &= \frac{1}{2}\left(
-\mu \pm \sqrt{\mu^2 + 4 \mu \mu'} \right), \label{r28}
\end{align}
where $\mu$ and $\mu'$ are given by Eqs.~(\ref{eq:mu}) and
(\ref{eq:muprime}).  The linear stability of a BFP is thus determined
by the signs of $\mu$ and $\mu'$ according to the following table.
\begin{center}
\begin{tabular}{l|cc}
& $\mu > 0$ & $\mu < 0$ \\
\hline
$\mu' > 0$ & SUS & UUU \\
$\mu' < 0$ & SSS & SUU
\end{tabular}
\end{center}
\end{theorem}

We next specialize this result to the dynamics on the sliding
surface.  
\begin{corollary}
\label{cor:restricted_bfp_stability}
The eigenvalues for a BFP of the dynamics Eq.~(\ref{eq:3DODE})
restricted to the sliding surface are given by $\lambda_{\pm}$ from
Eq.~(\ref{r28}).  The $xy$-projection of each of the corresponding
eigenvectors is proportional to $\ghat$.
\end{corollary}

The sliding surface stability information is summarized in Fig.~\ref{fig:stability_table}.
For each of the four stability types, the first two stabilities (in bold) describe the dynamics within the invariant sliding surface.
Equation~\ref{eq:muprime} places restrictions on the local convexity of the SZ at the BFP.
These possibilities are illustrated in Fig.~\ref{fig:stability_table}.
\begin{figure}[bt]
\includegraphics[width=1.0\linewidth]{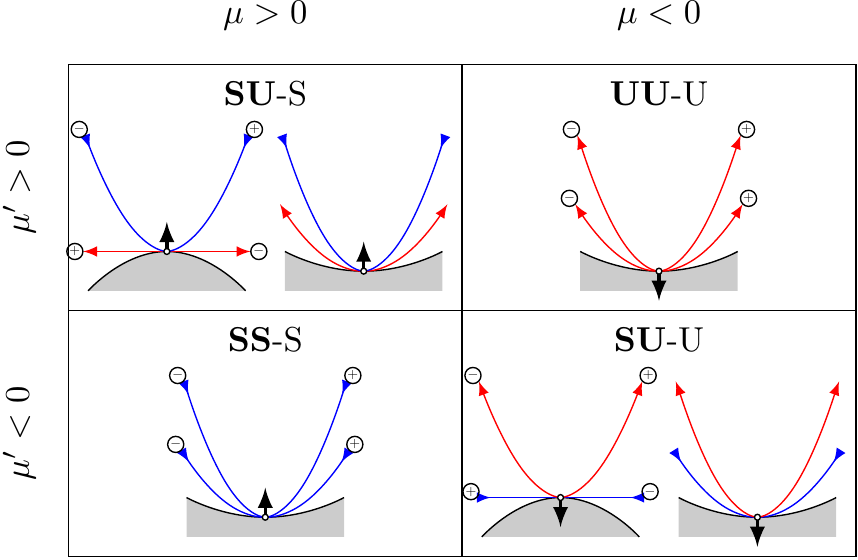}
\caption{(Color online.)  Four BFP stability types. Black arrow
  indicates burning direction. Gray regions are SZs. Two un/stable
  manifolds of each BFP are within the sliding surface. Each manifold
  is labeled with $+/-$ indicating its corresponding branch of the
  sliding surface.}
\label{fig:stability_table}
\end{figure}

\begin{lemma}
At an intersection $\mathbf{p}$ between a sliding front and the boundary of a SZ, the fluid flow is perpendicular to the sliding front. If the sliding front is tangent to the boundary, $\mathbf{p}$ is a BFP. Otherwise, $\mathbf{p}$ is a cusp along the sliding trajectory.
All BFPs and cusps of sliding fronts occur at the intersection between a sliding front and SZ boundary.
\label{lemma:SFIntersectsSZ}
\end{lemma}

\begin{proof}
  Combining the sliding constraint Eq.~(\ref{eq:sliding_constraint})
  with the front element dynamics Eq.~(\ref{eq:3DODE}) gives $\uvec
  \cdot \nhat = -v_0$. Since $|\uvec| = v_0$ on the SZ boundary,
  $\uvec = -v_0 \nhat$. Thus $\uvec$ is perpendicular to $\ghat$.

In Thm.~2 of Ref.~\cite{Mitchell12b}, it was shown that a necessary and sufficient condition for a BFP was for it to be on the boundary of the SZ with $\nhat$ perpendicular to the boundary. Thus $\ghat$ tangent to the boundary implies a BFP.

If $\ghat$ is not tangent to the boundary, then since the sliding trajectory cannot enter the SZ, it reaches the boundary and then must reverse direction forming a cusp.

Finally, a BFP and cusp both require $\rdot = 0$. 
This satisfies the sliding constraint and also implies $|\uvec| = v_0$.
\end{proof}

\begin{lemma}
A BFP or cusp on the boundary of the fluid domain, i.e. at a wall, must have $\ghat$ perpendicular to that boundary.
\label{lemma:SFOnFluidBoundary}
\end{lemma}

\begin{proof}
This follows from the previous Lemma~\ref{lemma:SFIntersectsSZ} and that the fluid velocity of an incompressible fluid is tangent to the fluid domain boundary.
\end{proof}

%\begin{lemma}
%At a tangency between a sliding front and a SZ, $\nhat \perp SZ$, there is a BFP, and thus the sliding front is a BIM \JM{rephrase because this might be stable or unstable BIM}. Also, burning away from SZ implies BFP is SSS or SSU, burning away from SZ implies BFP is UUU or SUU.
%\label{lemma:SFTangentToSZ}
%\end{lemma}
%
%\begin{proof}
%We mean by tangency that the projection of the sliding front is tangent to the SZ.
%So $\drdlambda$ tangent to $SZ$.
%By front compatibility, $\drdlambda \propto \ghat$, so $\nhat \perp SZ$---burning is either directly into or away from the SZ.
%By CITE prev paper, this is a BFP.
%A sliding front attached to BFP is a BIM.
%Previous result also implies that if BIM is burning away from SZ, BFP is SSS or SSU.
%If BIM is burning toward SZ, BFP is UUU or SUU. CHECK
%\end{proof}

%\JM{depending on how BIMsARESLIDINGFRONTS is proven, we may or may not keep this}
%Stable and unstable manifolds of fixed points of Eq.~\ref{eq:wpm_fields} are also stable and unstable manifolds of Eq.~\ref{eq:3DODE}.
%Therefore, BIMs are sliding fronts, and lie within the sliding surface.

An incompressible 2D fluid flow can be specified by a stream function
$\Psi(\rvec)$, with $u_x = d\Psi/dy, u_y = -d\Psi/dx$.  Each fluid
element follows a level set, or streamline, of $\Psi$.  Front
elements, on the other hand, do not follow streamlines, but generally
cross them one way or the other depending on their relative
orientation.
\begin{lemma}
Sliding fronts cross streamlines such that $\theta_{\nhat, \uvec}$, the angle between $\nhat$ and $\uvec$, satisfies $\cos(\theta_{\nhat, \uvec}) = - v_0 / |\uvec|$.
\label{lemma:SFsCrossStreamlinesAngle}
\end{lemma}
\begin{proof}
  From Eq.~(\ref{eq:3DODE}), $\uvec = \rdot - v_0 \nhat$. The sliding
  front condition implies $\uvec = \pm |\rdot| \ghat - v_0
  \nhat$. Dotting with $\nhat$, $\nhat \cdot \uvec = - v_0$.
\end{proof}
Lemma~\ref{lemma:SFsCrossStreamlinesAngle} means that sliding fronts are never tangent to streamlines, except in the $v_0 / |\uvec| \to 0$ limit.
Physical boundaries of the fluid (channel walls) are particularly important streamlines at which this lemma can be utilized.
%Also, at the intersection of two sliding fronts there is one quadrant that each burns toward.
%The angle between the sliding fronts in this quadrant is $2 \arcsin(v_0 / |\uvec|)$.

In addition to the angle at which sliding fronts cross streamlines, we can examine how rapidly they are crossed.
To this end, we calculate the rate at which $\Psi$ changes when viewed from the frame of an individual front element.
\begin{align}
\frac{D \Psi}{D t} &\equiv \frac{\partial \Psi}{\partial x} \frac{\partial x}{\partial t} + \frac{\partial \Psi}{\partial y} \frac{\partial y}{\partial t} + \frac{\partial \Psi}{\partial t}\\
&= -u_y \dot{x} + u_x \dot{y}\\
&= -(\dot{y} + v_0 \cos \theta) \dot{x} + (\dot{x} - v_0 \sin \theta) \dot{y}\\
&= -v_0 \rdot \cdot \ghat\\
&= - \sgn(\rdot \cdot \ghat) v_0 |\rdot|,
\end{align}
where the last equality makes use of the sliding constraint.  Scaling
by the front element speed, we find the simple relation
\begin{align}
\frac{D \Psi}{D s} = - \sgn(\rdot \cdot \ghat) v_0,
\label{eq:SFConstantChangeInPsi}
\end{align}
where $s$ measures the euclidean $xy$-length along the trajectory.
Equation~(\ref{eq:SFConstantChangeInPsi}) shows that the sliding
trajectories on the $+$ ($-$) branch of the sliding surface are those
that climb the stream function with the constant rate of descent
(ascent) $v_0$.  It is straightforward to show that the only front
elements that ascend or descend at a constant rate are sliding.
%\JM{I wanted to say something about constant energy injection, but it is hard to know what a length in phase space means. Is it true that in geostrophic flows that the stream function is proportional to sea surface height?}

%
\begin{figure}[bt]
\includegraphics[width=1.0\linewidth]{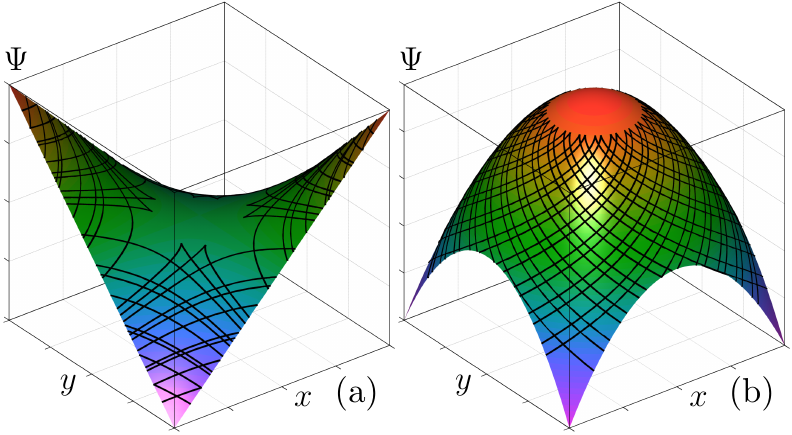}
\caption{(Color online.)  Sliding fronts represented on the graph of
  the stream function $\Psi$ near: (a) a hyperbolic point, (b) an
  elliptic point. Sliding fronts are curves of constant increase (or
  decrease) in the value of the stream function.}
\end{figure}

Consider two streamlines with values $\Psi_a$ and $\Psi_b$ and a
sliding front $F_{a,b}$ that connects one to the other with no
intervening cusps.  The $xy$-length of the segment $F_{a,b}$ follows
directly from Eq.~(\ref{eq:SFConstantChangeInPsi}),
\begin{align}
\label{eq:FFLength}
|F_{a,b}| &= \frac{|\Psi_b - \Psi_a|}{v_0}
\end{align}
This expression is particularly useful when thinking about FFs in
channel flows of arbitrary geometry.  Since the channel wall enforces
a boundary condition of constant $\Psi$, the FF length is found
through Eq.~(\ref{eq:FFLength}).  Thus the length of a FF that spans a
channel depends only on this ``energy difference'' between the two
walls and the burning speed, and not on other details of the flow.
While Eq.~(\ref{eq:FFLength}) was derived for a single sliding front, it
also holds for FFs that are composed of multiple BIM cores.  This can
be seen by applying Eq.~(\ref{eq:FFLength}) to each BIM segment
separately.  Interestingly, this implies that multiple FFs existing in
the same flow must have the same length, even in the absence of any
flow symmetry.

Equation~(\ref{eq:FFLength}) also implies that a channel of width $W$
cannot support a FF if $\Delta \Psi < v_0 W$.  Furthermore, if we
assume that a flow $\mathbf{u}$ \emph{without} wind gives no net flow
down the channel, then for the flow $\mathbf{u} + \mathbf{v}_w$,
Eq.~(\ref{eq:FFLength}) becomes
%
%the flow in the channel far from the FF location approaches a uniform
%wind of speed $v_w$,
%
\begin{equation}
|F_{a,b}| = W v_w / v_0.
\label{r30}
\end{equation}
This can be interpreted as the equality of fluid flux across the FF
and across the channel width.  Eq.~(\ref{r30}) shows that
$v_w \ge v_0$ is necessary but not sufficient.  In the case that FFs
do occur at $v_w = v_0$, they must be straight lines that meet the
channel walls at right angles.  This occurs exactly when the original
fluid flow ($v_w=0$) has a vertical advective separatrix.  This
condition is met for the windy alternating vortex chain model;
additionally Ref.~\cite{Schwartz08} experimentally demonstrated that
$v_w = v_0$ marked the onset of FFs.  However, it is not difficult to
construct flows where $v_w = v_0$ is not sufficient for the existence
of FFs.

\section{Stability of Frozen Domains}
\label{app:stability}

At the beginning of Sec.~\ref{sec:FF}, we specified that frozen
domains should be stable under small perturbations.  Here, we define
this stability more precisely.  For a given invariant burned domain,
with boundary $F$, we define an \emph{allowable} distortion of $F$ at
a point $\mathbf{r} \in F$ to be a distortion such that $F$ remains
unchanged outside a ball of radius $\epsilon(\mathbf{r})>0$ centered
at $\mathbf{r}$.  Note that the value of $\epsilon(\mathbf{r})$ is not
fixed but may vary with the point $\mathbf{r}$.
\begin{definition}[Stability of frozen domains/fronts]
  A frozen domain/front is required to be \emph{stable} in the
  following sense.  There must exist a function $\epsilon(\mathbf{r})$
  of each point $\mathbf{r}$ along the front $F$, i.e. the domain
  boundary, such that after any allowable distortion, the front
  remains pointwise close to $F$ and converges pointwise to $F$ as
  time goes to infinity.
\end{definition}
Here pointwise close is in the ``Lyapunov'' sense, in that the maximum
(over all time) distance from the time-evolved distorted front to the
original front $F$ remains bounded and goes to zero as $\epsilon$ goes
to zero.

With this definition, one can easily verify the argument of
Sec.~\ref{sec:FF} proving that a FF cannot contain an SUU or UUU BFP.
One also sees that a FF can contain an SSU or SSS BFP, and that the
curves constructed in Prop.~\ref{prop:FF} are stable.  Regarding the
latter, it is interesting to note how a perturbation of the FF $F$
returns to $F$.  First, a perturbation localized to the neighborhood
of an SSS point simply shrinks in size, back into the original FF, due
to the SSS point's being a sink.  Consider now a perturbation
localized at some point $\mathbf{r}$ of the FF that is not an SSS BFP.
This perturbation will be ``swept'' along the front, away from the
unstable BFP that generates the BIM; the perturbation might even
initially grow in size.  The localized perturbation will continue to
follow the BIM segment from which it was perturbed, and will
subsequently encounter either an SSS point, a domain wall, another BIM
segment of the FF, or it will be swept to infinity
(Fig.~\ref{fig:stability}).  In the initial three cases, it is clear
that the perturbation will disappear as it either shrinks into the SSS
BFP (Fig.~\ref{fig:stability}a), strikes the wall
(Fig.~\ref{fig:stability}b), or runs into the already burned region
(Fig.~\ref{fig:stability}c), assuming the initial size of the
perturbation $\epsilon(\mathbf{r})$ is sufficiently small.  The case
in which the perturbation is swept to infinity
(Fig.~\ref{fig:stability}d) requires an additional assumption on the
far field nature of the fluid flow, addressed in
Sec.~\ref{sec:stabilityFFinfinity}.

\begin{figure}[bt]
%\centering
%\includegraphics[width=1\linewidth]{../figures_and_code/stabilityfig.pdf}
\includegraphics[width=1\linewidth]{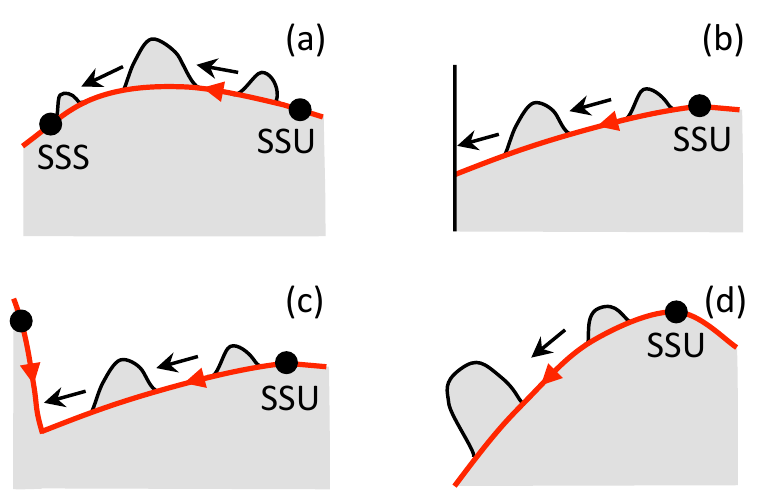}
\caption{ \label{fig:stability} (Color online.)  The fates of
  localized perturbations to a FF.  a) The perturbation eventually
  shrinks into an SSS BFP on the FF.  b) The perturbation strikes a
  wall of the fluid domain. c) The perturbation strikes a second BIM
  core at a concave corner of the FF.  d) The perturbation goes to
  infinity along a BIM core that stretches to infinity.  Though the
  size of the perturbation could grow indefinitely, the perturbed
  front still returns to the FF pointwise.}
\end{figure}
 
%%%%%%%%%%%%%%%%%%%%%%%%%%%%%%%%%%%%%%%%%%%%
\section{Dynamics at infinity}
\label{app:far_field}

In this paper, we assume that the fluid velocity field
$\mathbf{u}(\mathbf{r})$ is smooth and either (i) has a bounded stream
function or (ii) is ``localized'' with a simple far-field behavior.

More precisely, case (i) assumes that $\mathbf{u}(\mathbf{r})$ is
generated by a stream function $\Psi(\mathbf{r})$ that has a global
maximum and minimum over the fluid domain.  This case applies, for
example, to the stream function in Eq.~(\ref{eq:WAVC}), since in our
restricted domain $y$ is bounded between $0$ and $1$.  Under case (i),
no sliding front may be infinitely long without striking a SZ, since a
sliding front element increases or decreases its stream function value
at the constant rate $v_0$, Eq.~(\ref{eq:SFConstantChangeInPsi}).
Thus, for case (i) there are no BIM cores that stretch all the way to
infinity.

The remainder of this appendix concerns case (ii), stated more
precisely as the requirement that, as $r$ goes to infinity,
$\mathbf{u}(\mathbf{r})$ behaves as a homogeneous polynomial
$\mathbf{P}^k(x,y)$ of power $k \ge 0$ in the variables $(x,y)$, i.e.
\begin{equation}
\mathbf{u}(\mathbf{r}) = \mathbf{P}^k(x,y) + \mathbf{E}(x,y),
\label{r6}
\end{equation}
where $\mathbf{E}(x,y)$ is a term that grows more slowly than $r^k$.  More
precisely, we require that
\begin{equation}
\lim_{r \rightarrow \infty} r^{-k} \mathbf{E}(r \cos \phi, r\sin \phi)
= 0,
\label{r10}
\end{equation}
with the polar angle $\phi$ held constant in the limit.  Furthermore,
we require that radial derivatives of $\mathbf{E}$ grow more slowly,
according to
\begin{equation}
\lim_{r \rightarrow \infty} r^{-k+\ell} 
  \frac{\partial^\ell}{\partial r^\ell}\mathbf{E}(r \cos \phi, r\sin \phi)
= 0, \quad 0 \le \ell.
\label{r14}
\end{equation}
This assumption on the far field behavior eliminates, for example,
infinite arrays of vortices [though such arrays could be allowed under
case (i)], while allowing many important cases, including flows with
constant far field velocities, linear hyperbolic and elliptic flows,
other polynomial flows, and flows constructed from arbitrary
configurations of a finite number of vortices.

\subsection{Existence criteria for fixed points at infinity}

By a burning fixed point ``at infinity'', we intuitively mean a front
element trajectory that attains a constant orientation $\theta$ and
polar angle $\phi$ at an infinite value of $r$.  We can formalize this
definition by adapting the standard Poincar\'{e} compactification of
the plane~\cite{Perko01}.  First, map the radial distance $r$, $0 \le
r < \infty$, to a new radial variable $\rho$, $0 \le \rho < \pi/2$,
defined by
\begin{equation}
\rho = \arctan r. \label{r1}
\end{equation}
Similarly, introduce scaled Cartesian coordinates
${\bm \rho} = (\rho_x, \rho_y) = (x, y) \rho/r$.  This transformation
maps the $xy$-plane to an open disk of radius $\pi/2$ in the
${\bm \rho}$-plane.  We then include in our phase space the boundary
of the disk at $\rho = \pi/2$, representing the ``circle at
infinity''.  Equation~(\ref{r1}) yields
\begin{align}
\dot{\bm \rho} = 
& \frac{\rho}{\tan \rho}(\mathsf{I} - \hat{\bm \rho} \otimes \hat{\bm \rho} )
\dot{\mathbf{r}} +
(\cos^2 \rho)(\hat{\bm \rho} \otimes \hat{\bm \rho} ) \dot{\mathbf{r}}
\label{r3} \\
= &\frac{\rho}{\tan \rho}(\mathsf{I} - \hat{\bm \rho} \otimes \hat{\bm \rho} )
({\mathbf{u}} + v_0 \hat{\mathbf{n}})  \label{r7} \\
&+
(\cos^2 \rho)(\hat{\bm \rho} \otimes \hat{\bm \rho} ) ({\mathbf{u}} + v_0
\hat{\mathbf{n}}), \nonumber 
\end{align}
where the second equality follows from Eq.~(\ref{eq:3DODEa}).  Here,
$\mathsf{I}$ is the identity matrix and $\hat{\bm \rho} \otimes
\hat{\bm \rho}$ is the tensor product with components $(\hat{\bm \rho}
\otimes \hat{\bm \rho})_{ij} = \rho_i \rho_j /\rho^2$.  Equation
(\ref{r7}) is smooth in $(\rho_x, \rho_y, \theta)$ everywhere except
possibly at the boundary $\rho = \pi/2$.

Note that Eq.~(\ref{r6}) can be reexpressed in $(\rho, \phi)$
coordinates as
\begin{equation}
\mathbf{u}(\rho, \phi) = \cos^{-k} (\rho) \; \mathbf{Q}(\phi) + \bar{\mathbf{E}}(\rho,\phi),
\label{r8}
\end{equation}
where $\mathbf{Q}$ is smooth in $\phi$ and where
\begin{equation}
\lim_{\rho \rightarrow \pi/2} \cos^{k-\ell} (\rho)
\frac{\partial^\ell}{\partial \rho^\ell} \; \bar{\mathbf{E}}(\rho,\phi)
= 0, \quad 0 \le \ell,
\label{r11}
\end{equation}
with fixed $\phi$.  Thus as $\rho$ goes to $\pi/2$, Eq.~(\ref{r7})
scales as $1/\cos^{k-1} \rho$, and is thereby singular for $k > 1$.
To remove this singularity, the ODE time-parameter $t$ can be replaced
by a new scaled parameter $s$ defined by 
\begin{equation}
\frac{ds}{dt} = \frac{1}{\cos^{k-1} \rho}, \quad k \ge 1.
\label{r18}
\end{equation}
When $k=0$, we make no scaling.  This yields
%
%\begin{equation}
%\frac{d}{ds}{\bm \rho} = 
%\frac{\rho \cos^{k-1} \rho}{r}(\mathsf{I} - \hat{\bm \rho} \otimes \hat{\bm \rho} )
%\dot{\mathbf{r}} +
%\frac{\cos^{k-1} \rho}{1+r^2}(\hat{\bm \rho} \otimes \hat{\bm \rho} ) \dot{\mathbf{r}}.
%\label{r4}
%\end{equation}
%
\begin{align}
\frac{d}{ds}{\bm \rho} = &
\frac{\rho \cos^{k} \rho}{\sin \rho}(\mathsf{I} - \hat{\bm \rho} \otimes \hat{\bm \rho} )
 ({\mathbf{u}} + v_0 \hat{\mathbf{n}}) \label{r4} \\
& +
(\cos^{k+1} \rho)(\hat{\bm \rho} \otimes \hat{\bm \rho} ) ({\mathbf{u}} +
v_0 \hat{\mathbf{n}}), \quad k \ge 1.
\nonumber
\end{align}
The right-hand side of Eq.~(\ref{r4}) is now smooth over the entire
closed disk $\rho \le \pi/2$ and all $\theta$ values, except possibly
at $\rho = 0$, which does not concern us since our focus is at
$\rho = \pi/2$.  One can also see that
\begin{equation}
\frac{d}{ds}{\theta} = -(\cos^{k-1} \rho) \hat{n}_i
u_{i,j} \hat{g}_j, \quad k \ge 1, 
\label{r5}
\end{equation}
is smooth in $(\rho_x, \rho_y, \theta)$ for all $\rho < \pi/2$, but
might be singular at $\rho = \pi/2$.

We define a fixed point at infinity of Eq.~(\ref{eq:3DODE}) to be a
fixed point $(\rho_x^*, \rho_y^*, \theta^*)$ at $\rho^* = \pi/2$ for
Eqs.~(\ref{r4}) and (\ref{r5}) ($k \ge 1$) or Eqs.~(\ref{r3}) and
(\ref{eq:3DODEb}) ($k=0$).  We next restrict attention to those fixed
points at infinity that are the limits of sliding fronts extending to
infinity.  The tangent direction of a sliding front that converges
upon a fixed point at infinity must thus converge to
$\pm \hat{\bm \rho}$.  Thus, such a fixed point must yield
\begin{gather}
\hat{\mathbf{g}}^* = \pm \hat{\bm \rho},
\label{r13} \\
\lim_{r \rightarrow \infty} \frac{\dot{r}^2}{|\dot{\mathbf{r}}|^2} = 1.
\label{r15}
\end{gather}
Furthermore, notice that at $\rho = \pi/2$, the term
$(\cos^{k-1} \rho) {u}_{i,j} \hat{g}_j$ in Eq.~(\ref{r5}) will be
noninfinite if $\hat{\mathbf{g}} = \pm \hat{\bm \rho}$ [see
Eqs.~(\ref{r8}) and (\ref{r11})], a fortuitous consequence of
searching for fixed points that are the limits of sliding fronts.
 
Equation~(\ref{r15}) can be rewritten as
\begin{align}
1 &= \lim_{r \rightarrow \infty} \frac{(u_r + v_0 \hat{n}_r)^2}{|\mathbf{u}+v_0
  \hat{\mathbf{n}}|^2} 
= \lim_{r \rightarrow \infty} \frac{u_r^2}
{|\mathbf{u} \pm v_0  \hat{\bm \phi}|^2}  \label{r17} \\
&= \lim_{\rho \rightarrow \pi/2} \frac{\cos^{-2k}(\rho) Q_r^2(\phi)}
{\cos^{-2k}(\rho) Q_r^2(\phi) +  [\cos^{-k}(\rho) Q_\phi(\phi) \pm
  v_0]^2 } \\
&= \lim_{\rho \rightarrow \pi/2} \frac{ Q_r^2(\phi)}
{ Q_r^2(\phi) +  [Q_\phi(\phi) \pm v_0 \cos^k(\rho)]^2 },
\end{align}
where the second equality follows from Eq.~(\ref{r13}) and the third
from Eqs.~(\ref{r8}) and (\ref{r11}).  This implies that for a fixed
point at infinity (that is the limit of a sliding front) in polar
coordinates $(\pi/2, \phi^*, \theta^*)$
\begin{equation}
Q_\phi(\phi^*) = \pm v_0 \cos^k(\pi/2) = 
\begin{cases}
\pm v_0 & \text{ if } k = 0, \\
0 & \text{ if } k \ge 1.
\end{cases}
\label{r16}
\end{equation}

Assume now that $k \ge 1$.  We seek necessary and sufficient
conditions for $d {\bm \rho} / ds = 0$ at $\rho = \pi/2$.
Equations~(\ref{r8}), (\ref{r11}), and (\ref{r4}) imply that at
$\rho = \pi/2$
\begin{align}
\frac{d}{ds}{\bm \rho} (\pi/2, \phi, \theta) = 
\frac{\pi}{2} (\mathsf{I} - \hat{\bm \rho} \otimes \hat{\bm \rho} )
 {\mathbf{Q}}(\phi) = \frac{\pi}{2} \hat{\bm \phi} Q_\phi(\phi), 
\end{align}
which is independent of $\theta$.  Thus, a necessary and sufficient
condition for $d {\bm \rho} / ds = 0$ at $\rho = \pi/2$ is
\begin{equation}
Q_\phi(\phi^*) = 0,  \quad k \ge 1,
\label{r9}
\end{equation}
consistent with the sliding front condition Eq.~(\ref{r16}).  

We now seek necessary and sufficient conditions for
$d \theta/ ds = 0$.  Based on Eq.~(\ref{r5}), we define the scaled
Jacobian matrix
\begin{equation}
\mathsf{J}(\rho,\phi) = \cos^{k-1} \rho 
\left[
\begin{array}{cc}
u_{r,r} & u_{r,\phi}/r \\
u_{\phi,r} & u_{\phi,\phi}/r 
\end{array}
\right]
\end{equation}
expressed in the $(\hat{\mathbf{r}}, \hat{\bm \phi})$ basis.  Under
assumptions (\ref{r8}) and (\ref{r11}), one can show that $\mathsf{J}$
takes on the following form at $\rho = \pi/2$.
\begin{equation}
\mathsf{J}(\pi/2, \phi) = 
\left[
\begin{array}{cc}
k Q_r(\phi) & Q_{r,\phi}(\phi) \\
k Q_\phi(\phi)  & Q_{\phi,\phi}(\phi)
\end{array}
\right].
\end{equation}
For a fixed point at $\rho = \pi/2$, we apply Eq.~(\ref{r9}) to obtain
\begin{equation}
\mathsf{J}(\pi/2, \phi^*) = 
\left[
\begin{array}{cc}
k Q_r(\phi^*) & Q_{r,\phi}(\phi^*) \\
0  & Q_{\phi,\phi}(\phi^*)
\end{array}
\right].
\label{r12}
\end{equation}
Setting $d \theta/ d s= 0$ in Eq.~(\ref{r5}) is equivalent to
$\hat{\mathbf{g}}$ being an eigenvector of $\mathsf{J}$.  From
Eqs.~(\ref{r13}) and (\ref{r12}), however, we see that
$\hat{\mathbf{g}} = \hat{\bm \rho}$ is already guaranteed to be an
eigenvector of $\mathsf{J}(\pi/2,\phi^*)$.  Thus, $d \theta/ d s= 0$
is already implied by Eqs.~(\ref{r13}) and (\ref{r9}).

Assume now that $k=0$.  Then $\mathbf{P}^k(x,y)$ is constant, and it
is easy to see from Eq.~(\ref{r7}) and Eq.~(\ref{eq:3DODEb}) that
$\dot{\bm \rho} = 0$ and $\dot{\theta}=0$ at $\rho = \pi/2$ for any
$\phi$ and $\theta$.  However, $\phi$ and $\theta$ are no longer
arbitrary when we apply the sliding front conditions Eq.~(\ref{r13})
and Eq.~(\ref{r16}).  

To summarize, we have the following conditions on the existence of
fixed points at infinity.
\begin{theorem}
\label{thm:BFP}
Assume the incompressible flow $\mathbf{u}$ satisfies Eqs.~(\ref{r6})
and (\ref{r14}) or equivalently Eqs.~(\ref{r8}) and (\ref{r11}).  Then
a sliding front ends in a fixed point at infinity with orientation and
position angles $\theta^*$ and $\phi^*$ if and only if
$\hat{\mathbf{g}}(\theta^*) = \pm \hat{\bm \rho}(\phi^*)$ and
\begin{align}
Q_\phi(\phi^*)=0 \text{, for } k \ge 1, \\
Q_\phi(\phi^*)=\pm v_0 \text{, for } k = 0.
\end{align}
\end{theorem}

\subsection{Stability of fixed points at infinity}

Assume $k \ge 1$.  To determine the stability of fixed points at
infinity, we compute the Jacobian matrix of the flow (\ref{r4}) and
(\ref{r5}) in $(\rho, \phi, \theta)$ coordinates, i.e.
\begin{equation}
\mathcal{J} = \frac{\partial (\rho', \phi', \theta')}{\partial (\rho,
  \phi, \theta)}.
\end{equation}
Applying the results of Theorem~\ref{thm:BFP}, it can be shown that
the Jacobian at a fixed point at infinity (to which a sliding front
converges) is
\begin{equation}
\mathcal{J} = 
\left[
\begin{array}{ccc}
-Q_r & 0 & 0 \\
\phi'_{,\rho}   & Q_{\phi,\phi} & 0 \\
0 & k Q_{\phi,\phi} & Q_{\phi,\phi} - k Q_r
\end{array}
\right],
\end{equation}
where
\begin{equation}
\phi'_{,\rho}  = -k(\cos \rho)^{k-1} ( \bar{E}_\phi + v_0 \hat{n}_\phi). 
\end{equation}

Applying the area-preservation constraint $\nabla \cdot {\mathbf{u}} = 0$ to
Eqs.~(\ref{r8}) and (\ref{r11}), we find
\begin{equation}
Q_{\phi,\phi} = -(1+k)Q_r,
\end{equation}
from which 
\begin{equation}
\mathcal{J} = 
\left[
\begin{array}{ccc}
-Q_r & 0 & 0 \\
\phi'_{,\rho}   & -(1+k)Q_r & 0 \\
0  & -(1+k)k Q_r & -(1+2k)Q_r
\end{array}
\right].
\end{equation}
Thus, the eigenvalues of $\mathcal{J}$ all have the same sign, which
is the opposite sign to $Q_r$.  We thus have the following.

\begin{theorem}
  Assume the incompressible flow $\mathbf{u}$ satisfies
  Eqs.~(\ref{r6}) and (\ref{r14}), with $k \ge 1$.  Suppose a sliding
  front ends in a fixed point $(\pi/2, \phi^*, \theta^*)$ at infinity
  and assume $Q_r(\phi^*) \ne 0$.  Then the fixed point has stability
  SSS (UUU) if the flow along the sliding front is radially outward
  (inward), i.e. $Q_r(\phi^*) > 0$ ($Q_r(\phi^*) < 0$).
\end{theorem}

In the $k = 0$ case, the eigenvalues of the Jacobian are all 0, and
stability must be determined by nonlinear analysis. 

\subsection{Behavior of frozen fronts at infinity}

Consider a sliding front that extends to infinity without containing a
cusp beyond some set radius.  Such a front must be a trajectory of the
$\mathbf{w}_\pm$ field Eq.~(\ref{eq:wpm_fields}) with a single choice
of sign.  Transforming this vector field into the
${\bm \rho}$-coordinates according to Eq.~(\ref{r3}) and scaling time
via Eq.~(\ref{r18}), we obtain a smooth 2D vector field over
${\bm \rho}$-space in the neighborhood of the boundary $\rho = \pi/2$.
Thus the sliding front that extends to infinity, i.e. converges to
$\rho = \pi/2$ in the ${\bm \rho}$-coordinates, can do only one of two
things.  If the boundary $\rho = \pi/2$ contains no fixed point, then
the boundary is a limit cycle to which the sliding front converges.
Otherwise, the sliding front must converge upon a fixed point at
$\rho = \pi/2$.

Let us suppose that the boundary $\rho = \pi/2$ is a limit cycle,
which we also suppose is stable.  (The analysis of the unstable case
is similar.)  Then, a circle at $\rho = \pi/2 -\epsilon$ will converge
outward to $\pi/2$ for a small enough $\epsilon$.  In the original
$xy$-coordinates, this means that a large enough circle of front
elements would expand outward without bound.  If these front elements
were passive tracers of the flow, this outward expansion would clearly
violate conservation of area of the underlying flow.  Similarly, if
these front elements were all pointing inward with nonzero $v_0$, they
would clearly still violate conservation of area of the underlying
flow.  Thus, the front elements can only be facing outward if the
circle is ultimately to grow in size.  In summary, we have the
following.

\begin{proposition}
  Assume the incompressible flow $\mathbf{u}$ satisfies
  Eqs.~(\ref{r6}) and (\ref{r14}).  A sliding front that moves outward
  to (inward from) infinity, without cusps beyond a certain radius,
  with its burning direction pointed inward (outward) must converge
  upon a fixed point $(\pi/2, \phi^*, \theta^*)$ at infinity.
\end{proposition}

Consider now a frozen front that converges on a limit cycle at
infinity.  There would in fact need to be two separate frozen fronts
converging upon a limit cycle at infinity, one burning inward and one
burning outward, to create a burned strip that spirals outward to
infinity.  However, the previous proposition shows that only one such
burning direction is possible.  Hence, we have shown the following.

\begin{proposition}
  Assume the incompressible flow $\mathbf{u}$ satisfies
  Eqs.~(\ref{r6}) and (\ref{r14}).  Any frozen front that extends to
  infinity, without any concave corners or BFPs beyond a certain
  radius, must be a sliding front converging upon a fixed point at
  infinity.
\end{proposition}

\subsection{Stability of frozen fronts extending to infinity}

\label{sec:stabilityFFinfinity}

We now complete the discussion of frozen front stability begun in
Appendix~\ref{app:stability}.  In particular we consider a FF $F$ that
extends to infinity and an allowable perturbation to $F$ at a point
$\mathbf{r} \in F$ that is, at least initially, swept toward infinity.
Our concern is that this perturbation may not only grow in time, but
could burn far enough away from $F$ that it would no longer be carried
away to infinity but would leave some new part of the fluid, not in
the original burned domain defined by $F$, burned forever.  This
problem is resolved, however, by the Poincar\'{e} compactification.
If the sliding front converges on a stable fixed point at
$\rho = \pi/2$, then one can always find a sufficiently small
perturbation $\epsilon(\mathbf{r})$ such that the perturbation never
grows too large to be collapsed by the SSS BFP at infinity.

%%%%%%%%%%%%%%%%%%%%%%%%%%%%%%%%%%%%%%%%%%%%
\section{SSS points}
\label{app:SSS}

It was shown in Ref.~\cite{Mitchell12b} that SSS is a possible stability type for BFPs.
Here we provide an explicit example for the fluid flow surrounding an SSS point.

The stream function
\begin{align}
\Psi = - v_0 y - \mu x y + \frac{c \mu^2}{6 v_0} y^3
\end{align}
produces the following flow field with a BFP at $(x,y) = (0,0)$ and
$\hat{\mathbf{g}} = \hat{\mathbf{y}}$
\begin{align}
\label{eq:SSS}
u_x &= - v_0 - \mu x + \frac{c \mu^2}{2 v_0} y^2,\\
u_y &= \mu y,
\end{align}
where $v_0$ is the burning speed, $\mu$ is the eigenvalue from
Eq.~(\ref{eq:mu}), and $c$ is a new parameter.  Note that this flow
can be understood as the combination of a uniform ``wind'', a linear
hyperbolic flow, and a Poiseuille flow.  The stability of the BFP can
be determined from Theorem~\ref{t4} once $\mu'$ is computed.
Ref.~\cite{Mitchell12b}, Eqs.~(19) and (31), showed that
\begin{align}
\mu' &= (2 \mu^2 - v_0 \mathbf{a} \cdot \mathbf{g})/ \mu, \\
\mathbf{a} \cdot \mathbf{g} &= n_i g_j g_k u_{i,jk},
\end{align}
from which it is straightforward to compute
$\mathbf{a} \cdot \mathbf{g} = c \mu^2/v_0$ and $\mu' = (2-c) \mu$.
Thus, the BFP at the origin has stability SSS if $\mu > 0$ and
$c > 2$.  If in addition $c \le 9/4$, Eqs.~(\ref{r27}) and (\ref{r28})
show that the three eigenvalues are real.

Figure~\ref{fig:SSS_not} illustrates this flow for $c = 7/4$, which is
too small to produce an SSS BFP at the origin.
Equation~(\ref{eq:muprime}) allows us to interpret this fact as due to
the magnitude of the SZ curvature $|\kappa|$ being too small.
Topologically, the dynamics has the same structure as the linear
hyperbolic flow Fig.~\ref{fig:hyperbolic_wpm_surface_and_flow}, since
the origin is an SSU BFP.  The BIM (red) forms a FF extending to
infinity.

\begin{figure}[bt]
\includegraphics[width=\linewidth]{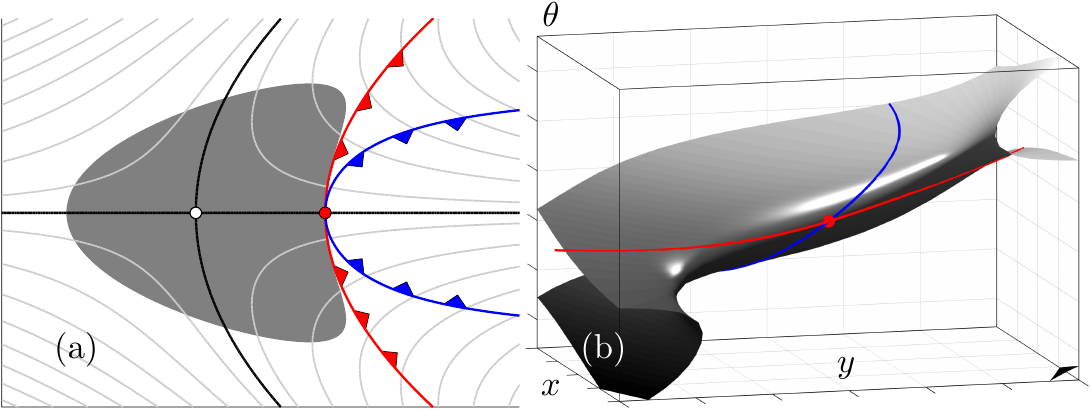}
\caption{(Color online.)  When the parameter $c$ is below $2$, the
  curvature of the SZ is not strong enough to create the SSU to
  SSU-SSS-SSU bifurcation. ($v_0 = 1, \mu = 1, c = 7/4$). (a) SSU BFP
  with attached sliding fronts.  Three other BFPs on the SZ boundary
  are not shown. (b) Fronts on the sliding surface seen in 3D from the
  center of the SZ.}
\label{fig:SSS_not}
\end{figure}

Increasing $c$ to $8.5/4$, the magnitude of the SZ curvature at the
BFP is increased (Fig.~\ref{fig:SSS_real}).  The SSU BFP in
Fig.~\ref{fig:SSS_not} has bifurcated into an SSS BFP at the origin
surrounded by two SSU BFPs.  The BIMs emanating from the SSU BFPs
toward the central SSS BFP actually terminate on the SSS BFP (most
easily seen in Fig.~\ref{fig:SSS_real}b.)  Thus, all three BFPs lie on
the FF composed of these BIMs.
\begin{figure}[bt]
\includegraphics[width=\linewidth]{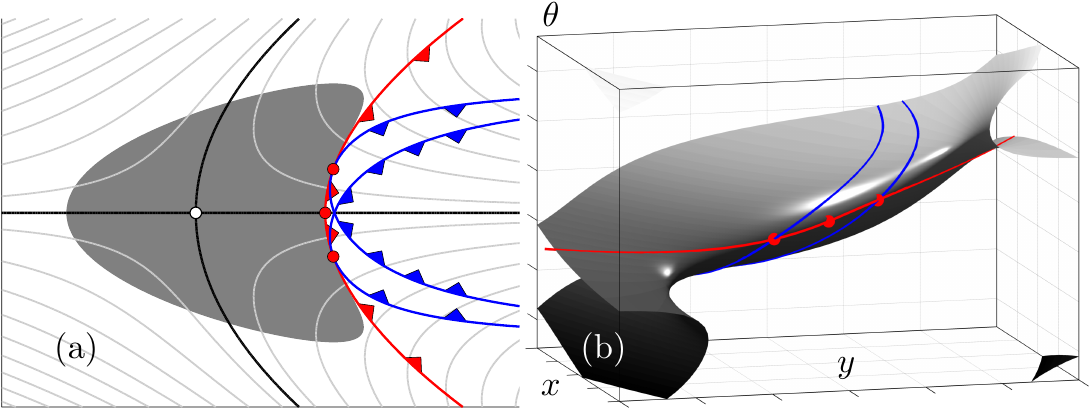}
\caption{(Color online.)   ($v_0 = 1, \mu = 1, c = 8.5/4$). (a) When the central SSS BFP has real eigenvalues, the BIMs emanating from the two neighboring SSU BFPs reach it with no cusps. (b) BIMs run very close to the sliding surface fold.}
\label{fig:SSS_real}
\end{figure}

Increasing $c$ further to $15/4$, the SSS point acquires complex
eigenvalues.  Therefore, any sliding trajectory that reaches the
center point must first encircle it an infinite number of times
(Fig.~\ref{fig:SSS_complex}).  Because the SSS point is on the fold of
the sliding surface (where $|\uvec| = v_0$), such a trajectory must
pass through the fold an infinite number of times, generating an
infinite sequence of cusps.  Therefore, for any point along a sliding
trajectory (which we imagine to be a part of some FF), there must
exist an infinite number of cusps between it and the SSS point.  This
ensures that an SSS point with complex eigenvalues cannot lie on a FF.

\begin{figure}[bt]
\includegraphics[width=\linewidth]{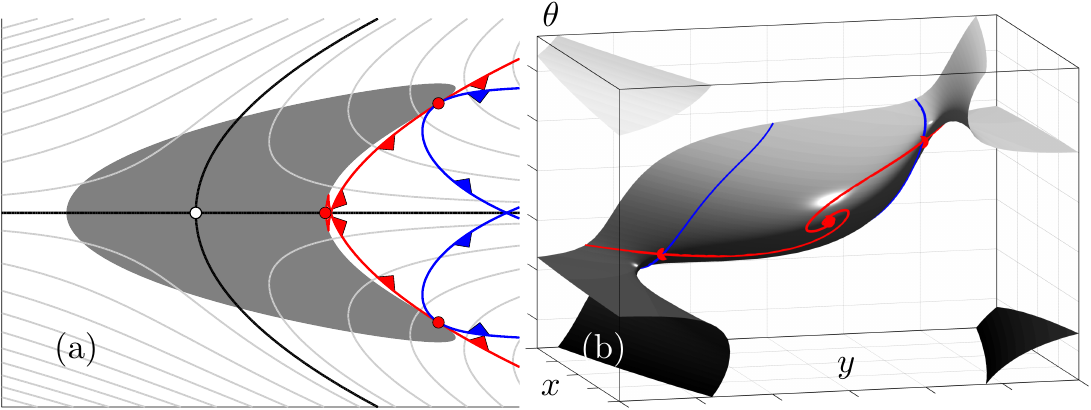}
\caption{(Color online.)  ($v_0 = 1, \mu = 1, c = 15/4$). (a) When the
  central SSS BFP has complex eigenvalues, the BIMs emanating from the
  SSU BFPs repeatedly overshoot the SSS BFP, and form cusps. (b) From
  side view, it is clear that these cusps are caused by the local
  spiraling of trajectories.}
\label{fig:SSS_complex}
\end{figure}

Conversely, an SSS point can only lie on a FF if its eigenvalues are
real.  Such SSS points do exist in the flow Eq.~(\ref{eq:SSS}) for
parameters satisfying $ v_0 > 0, \mu > 0, 8/4 < c < 9/4 $.  Thus in
Fig.~\ref{fig:SSS_real} the real eigenvalues ensure that the SSS point
lies on the FF attached to the neighboring SSU points.  Because the
eigenvalues are real for an open interval in parameter space, the
connection between the SSU and SSS BFPs in Fig.~\ref{fig:SSS_real} is
structurally stable.  Here the FF runs very close to the SZ boundary,
however there is a small gap.  Constrast this with the SSU to SSU
connection previously shown in Fig.~\ref{fig:SSU_to_SSU_connection},
which was not structurally stable.

\bibliographystyle{apsrev4-1}
%\bibliography{mybibKM}

%merlin.mbs 2010-03-15 4.21a (PWD, AO, DPC)

\end{document}